\documentclass[
  a4paper,
  twocolumn,
  11pt,
  unpublished
]{quantumarticle}
\pdfoutput=1

\usepackage[utf8]{inputenc}
\usepackage[english]{babel}
\usepackage[T1]{fontenc}
\usepackage{amsmath}
\usepackage{amssymb}       
\usepackage{mathtools}     
\usepackage{nicefrac}
\newcommand{\defeq}{\mathrel{\mathop:}=} 
\usepackage{bm}            
\usepackage{amsthm}

\usepackage{graphicx}
\usepackage[dvipsnames]{xcolor}

\usepackage{physics}

\newcommand{\sharpP}{$\#\mathsf{P}$}
\newcommand{\Poly}{$\mathsf{P}$}
\newcommand{\NP}{$\mathsf{NP}$}
\newcommand{\BQP}{\mathsf{BQP}}
\newcommand{\BPP}{\mathsf{BPP}}

\usepackage[numbers,sort&compress]{natbib}

\usepackage{hyperref}
\usepackage[nameinlink,noabbrev]{cleveref}

\definecolor{violet}{HTML}{53257F}     
\definecolor{green}{HTML}{257a7f}      
\definecolor{blue}{HTML}{254280}
\definecolor{orange}{HTML}{ED6F1C}
\definecolor{brightgreen}{HTML}{C2ED1C}
\definecolor{fuchsia}{rgb}{1.0, 0.0, 1.0}

\hypersetup{
  colorlinks=true,
  linkcolor=violet,
  citecolor=violet,
  urlcolor=violet,
  allbordercolors=violet      
}

\makeatletter
\newcommand{\subscripts}[3]{%
  \@mathmeasure\z@\displaystyle{#2}%
  \global\setbox\@ne\vbox to\ht\z@{}\dp\@ne\dp\z@
  \setbox\tw@\box\@ne
  \@mathmeasure4\displaystyle{\copy\tw@_{#1}}%
  \@mathmeasure6\displaystyle{{#2}_{#3}}%
  \dimen@-\wd6 \advance\dimen@\wd4 \advance\dimen@\wd\z@
  \hbox to\dimen@{}\mathop{\kern-\dimen@\box4\box6}%
}
\makeatother

\newcommand{\supp}{\textup{\text{supp}}}
\newcommand{\se}{\subseteq}

\newcommand{\gs}{\geqslant}
\newcommand{\sm}{\setminus}

\newtheorem{theorem}{Theorem}
\newtheorem{lemma}{Lemma}
\newtheorem{proposition}{Proposition}

\newtheorem{definition}{Definition}
\newtheorem{remark}{Remark}
\newtheorem{conjecture}{Conjecture}
\newtheorem{corollary}{Corollary}
\newtheorem{result}{Result}

\newtheorem*{theorem*}{Theorem}
\newtheorem*{corollary*}{Corollary}
\newtheorem*{proposition*}{Proposition}

\begin{document}

\title{The Structure of Circle Graph States}

\author{Frederik Hahn}
\email{mail@frederikhahn.eu}
\affiliation{Technical University of Berlin, Electrical Engineering and Computer Science Department, Berlin, Germany}

\author{Rose McCarty}
\affiliation{Georgia Institute of Technology, School of Mathematics and School of Computer Science, USA}

\author{Hendrik Poulsen Nautrup}
\affiliation{University of Innsbruck, Department of Theoretical Physics, Technikerstra\ss e 21a, A-6020 Innsbruck, Austria}

\author{Nathan Claudet}
\affiliation{University of Innsbruck, Department of Theoretical Physics, Technikerstra\ss e 21a, A-6020 Innsbruck, Austria}

\begin{abstract}
Circle graph states are a structurally important family of graph states.
The family's entanglement is a priori high enough to allow for universal measurement-based quantum computation (MBQC); however, MBQC on circle graph states is actually efficiently classically simulable.
In this work, we paint a detailed picture of the local equivalence of circle graph states.
First, we consider the class of all graph states that are local unitary (LU)-equivalent to circle graph states.
In graph-theoretical terms, this LU-equivalence class is the set of all graphs reachable from the family of circle graphs by applying $r$-local complementations.
We prove that the only graph states that are LU-equivalent to circle graph states are circle graph states themselves: circle graphs are closed under $r$-local complementation.
Second, we show that bipartite circle graph states, i.e., 2-colorable circle graph states, are in one-to-one correspondence with
planar code states, on which MBQC is known to be efficiently classically simulable.
Leveraging this correspondence, we present alternative, simple proofs that (1) if a planar code state is LU-equivalent to a stabilizer state, they are in fact local Clifford (LC)-equivalent to it and that (2) MBQC on all circle graph states is efficiently classically simulable.
Third and finally, we demonstrate that the problem of counting the number of graph states LU-equivalent to a given graph state is \sharpP-hard.
\end{abstract}

\maketitle

\section{Introduction}

While it is believed that quantum computers can solve some problems faster than classical computers, the source of this potential speedup is still far from well understood. It is known, however, that we need at least two ingredients: quantum entanglement, and non-Clifford operations \cite{gottesman1998heisenberg, jozsa2003role, nielsen2010quantum}.

In measurement-based quantum computation (MBQC) these two sources of quantum computational power are clearly separated.
The key idea behind MBQC is to prepare an entangled resource using only Clifford operations, and then to perform a computation using only single-qubit measurements, along with classical feedback and control \cite{Raussendorf2001,raussendorf2003measurementbased,nautrup2024measurement}.
During this measurement-based computation, no further entanglement is added to the initial resource.
A central aspect of MBQC is to determine which resource families allow for universal quantum computation, that is, which families provide as much computational power as any quantum computer~\cite{VandenNest2007}. To answer this question research has drawn upon and cross-pollinated various fields of study, from condensed-matter physics~\cite{raussendorf2019computationally,stephen2019subsystem} to graph theory~\cite{hein2004multiparty,mhalla2012graph}. The first known universal resource for MBQC was the cluster state~\cite{briegel2001}, where qubits are placed on nodes of a two-dimensional grid and Ising-type interaction are applied along edges. Its natural generalization are graph states, where qubits are placed on vertices of arbitrary graphs and Ising-type interactions are applied along edges~\cite{hein2004multiparty,schlingemann2001quantum}.

As of today, there is no known systematic way of telling which families of graph states allow for universal MBQC, or when MBQC is efficiently classically simulable on such a family.
However, what is known is that to allow for universal MBQC, the entanglement of a resource family needs to be complex enough.
More precisely, the family's measures of entanglement such as the  entanglement width and Schmidt rank-width~\cite{van2006universal, shi2006classical}, should grow quickly with the number of qubits.
For graph states, both these entanglement measures coincide and correspond to the well-known graph-theoretic rank-width (see~\cite{rankWidthSurvey} for a survey on rank-width).
If the rank-width of a family of graphs grows at most logarithmically with the number of graph vertices $n$, MBQC on the corresponding graph states is efficiently classically simulable~\cite{vandennestClassicalSimulationUniversality2007}.

For example, graph states whose underlying graphs are  distance-hereditary (i.e.~graphs of rank-width at most 1 such as trees) do not allow for universal MBQC, whereas graph states corresponding to the $n\times n$ grid are universal MBQC resources (with their rank-width growing as $n-1$, see \cite{jelinekRankwidthSquareGrid2010}).

Graph states arising from the so-called \textit{circle graphs} --which we call \textit{circle graph states} in this work-- are a priori good candidates for MBQC universality.
Indeed for circle graphs the rank-width grows faster than logarithmically.
We prove in Appendix~\ref{app:rank_width_of_circle_graphs} that $n \times n$ comparability grids, a special type of circle graph, have rank-width at least $n/4$.

Circle graphs and comparability grids arise naturally as an obstruction to having small rank-width. Informally, a graph state has ``large'' rank-width if and only if a ``large'' comparability grid graph state can be obtained from it using single-qubit local operations and classical communication (LOCC)~\cite{geelenGridTheoremVertexminors2023,transformingStatesOther, vandennestGraphicalDescriptionAction2004}.

These single-qubit operations are so-called \textit{local} operations and do not change universality or simulability of families of quantum states.
Thus, a statement about the universality or simulability of family of quantum states, holds for the entire orbit of the family under local operations.

Remarkably, the orbits of graph states under reversible local operations can be neatly described using graphical transformations of the underlying graphs.  For local Clifford operations, a subset of local unitary operations, the corresponding graph transformations are \textit{local complementations}~\cite{vandennestGraphicalDescriptionAction2004, heinEntanglementGraphStates2006}.
For arbitrary local unitary operations, the corresponding graph transformations referred to as \textit{$r$-local complementations}~\cite{claudet2024local}.

In this work, we characterize the transformation of circle graph states under arbitrary reversible local unitary operations (that is, the transformation of circle graphs with respect to the generalized $r$-local complementation) and illustrate their connection to planar code states. This gives another reason why MBQC on circle graph states is classically simulable~\cite{harrison2025fermionicinsightsmeasurementbasedquantum}.

\subsection{Results}

The family of circle graph states is known to be closed under local Clifford (LC) operations, since circle graphs are closed under local complementations.
This means that all graph states that are LC-equivalent to circle graph states are circle graph states themselves.
We show that this extends to local unitary (LU)-equivalence.

\begin{result}[see \Cref{thm:circle_graphs_closed}]
    For circle graph states $\text{LU} = \text{LC}$.
\end{result}

The only graph states that are LU-equivalent to circle graph states are circle graph states themselves.
In other words, circle graphs are closed under $r$-local complementation.

We proceed to demonstrate that bipartite circle graph states correspond to the well-known CSS states defined on a planar graph, the so-called planar code states.

\begin{result}[see \Cref{thm:planar_code_states_are_bipartite_circle_graph_states}]
    Every planar code state is LC-equivalent to a bipartite circle graph state. Conversely, every bipartite circle graph state is LC-equivalent to a planar code state.
\end{result}

Several significant results from the theory of graph states can be directly derived as corollaries of our findings, either recovering or even improving upon them.

For example, the above theorems imply that every stabilizer state that is LU-equivalent to a planar code state is also LC-equivalent to it.

\begin{result}[see \Cref{cor:planar_code_states_LU_LC}]
    For planar code states $\text{LU} = \text{LC}$.
\end{result}

Note that the above result was already known for planar code states under some additional restrictions \cite{Sarvepalli2010PRA}.

We further show  that every circle graph state can be obtained from some bipartite circle graph state with a qubit number that is at most quadratically larger.

\begin{result}[see \Cref{prop:vertex_minor_of_bipartite_circle_graph}]
  Every circle graph is a vertex-minor of some bipartite circle graph with a quadratic overhead.
\end{result}

Since circle graph states can be obtained from bipartite circle graph states with only a polynomial overhead, and it is also known that MBQC on planar code states is efficiently classically simulable~\cite{Bravyi2007},
we can conclude that MBQC on circle graph states is efficiently classically simulable.

\begin{result}[see \Cref{cor:MBQC_ECS_on_circle_graph_states}]
    MBQC on circle graph states is efficiently classically simulable.
\end{result}

Note that the above result was already known through a connection with fermionic Gaussian states~\cite{harrison2025fermionicinsightsmeasurementbasedquantum}.

The polynomial rank-width scaling of graphs is a prerequisite for the efficient universality of the corresponding graph state resources~\cite{vandennestClassicalSimulationUniversality2007}.
Despite being efficiently classically simulable, we find that circle graphs have polynomial rank-width.

\begin{result}
[see \Cref{corollary:circle_graphs_rank-width}]
There exist infinitely  many $n$-vertex circle graphs with rank-width $\Omega(\sqrt{n})$.
\end{result}

Since circle graphs have at least polynomial rank-width, we then recover the fact that while polynomial rank-width is a necessary condition for efficient universality, it is not a sufficient one~\cite{Bravyi2007}.

\begin{result}[see \Cref{cor:rank-width_sufficient_condition}]
Assuming $\BPP \neq \BQP$, polynomial rank-width is not a sufficient condition for efficient universality of a graph state resource for MBQC.
\end{result}

Finally, the problem of counting the number of graph states LC-equivalent to a given graph state is \sharpP-complete, even when restricted to circle graph states~\cite{dahlberg2020counting}.
We can generalize this statement to encompass general LU-equivalence.

\begin{result}[see \Cref{cor:counting}]
    Counting the number of graph states LU-equivalent to a given graph state is \sharpP-hard, even when restricted to circle graph states.
\end{result}

In the following section, we provide more context for our findings, suggest future research directions, and highlight open problems that we believe warrant further exploration.

\subsection{Discussion}

\subsubsection{Graph states representation of general surface codes}

In this work, we consider the correspondence of bipartite circle graphs to planar code states, that is, to the codewords of the Kitaev surface code \cite{bravyi1998quantumcodeslatticeboundary, Kitaev2003} defined on a sphere.
In general, it is not clear what graph states correspond to codewords of surface code states, where the surface has a non-zero genus.
However, in the special case where the surface code is defined on a grid embedded in a torus, the graph state corresponding to the codewords is known \cite{Liao2021}.
We leave the characterization of (bipartite) graph states corresponding to general surface codes for future work.
Such a characterization may lead to new insights into the classical simulability of graph states beyond circle graphs, by relating them to higher-genus surface-code states, for which MBQC simulation costs scale polynomially in the system size but (in general) exponentially in the genus~\cite{Goff2012}.

\subsubsection{Closure of circle graph states under SLOCC}

As circle graph states are known to be closed under taking vertex-minors, if a graph state is obtained from a circle graph state by local Clifford operations, local Pauli measurements and classical communication (LC+LPM+CC), then it is a circle graph state itself~\cite{bouchetCircleGraphObstructions1994, dahlbergComplexityVertexminorProblem2022}.
In this work, we do prove that if a graph state is obtained from a circle graph state by any reversible local unitary transformation, then it is a circle graph state.
However, it remains unknown whether any graph state obtained from a circle graph state through the more general stochastic local operations and classical communication (SLOCC) will necessarily be a circle graph state.
A result like this would be an even stronger statement about the universality, or lack thereof, of circle graphs.

More precisely, MBQC on circle graph states being efficiently classically simulable implies that, under the assumption that $\BQP \neq \BPP$ (i.e. quantum computers are more powerful than classical probabilistic computers), circle graph states are not efficiently universal. Informally, a family of quantum states  is efficiently universal if any quantum state that can be created efficiently with a quantum circuit, can be created efficiently by LOCC from a state in the family with polynomial overhead in the number of qubits \cite{VandenNest2007}. Alternatively, a family is efficiently universal if any grid graph state can be prepared efficiently by LOCC from a state in the family with a polynomial overhead in the number of qubits.

As is, MBQC on circle graph states being efficiently classically simulable does not rule out that circle graph states are universal, in the sense of classical-quantum (CQ) universality \cite{VandenNest2007}. A family of quantum states is universal, or CQ-universal, if any quantum state can be obtained by LOCC from a state in the family (without the concern of efficiency).
We leave open the question of whether circle graph states are (CQ) universal.
However, we remark that the opposite would be directly implied by circle graph states being closed under SLOCC (and thus LOCC).

\subsubsection{Unifying the proofs of efficient simulation}

The now two known proofs that MBQC on circle graph states is efficiently classically simulable stem from two distinct interpretation of circle graphs.
First, bipartite circle graphs are the fundamental graphs of planar graphs, which is why they are connected to planar code states (see \Cref{sec:planar_code_states}).
Second, circle graphs can be interpreted  as Eulerian tours of 4-regular multigraphs, leading to the connection with fermionic Gaussian states \cite{harrison2025fermionicinsightsmeasurementbasedquantum}.
Finding unifying principles that explain when quantum systems are efficiently classically simulable is of interest (see, for example, \cite{bauer2026quadratictensorsunificationclifford}),
and circle graph states may constitute a key toy example.

\subsubsection{Geometric measure of entanglement of circle graph states}

As mentioned in the introduction and shown in Appendix \ref{app:rank_width_of_circle_graphs}, the fact that MBQC on circle graphs is efficiently classically simulable is not due to a lack of entanglement.
However, while a lack of sufficient entanglement forbids universality for MBQC, it is also the case that too high of a so-called \textit{geometric measure of entanglement} forbids universal quantum computation.

This geometric measure has been used to prove that almost all quantum states are not universal resources for MBQC \cite{Gross2009}.
Indeed, when the geometric measure of entanglement is too high (that is, when the deviation of the geometric measure from $n$ is logarithmic in $n$, where $n$ is the number of qubits), then MBQC can essentially be simulated with coin flips.

The geometric measure of entanglement of almost all graph states is not high enough to apply the same argument \cite{Ghosh2025}; indeed, the deviation of the geometric measure from $n$ of a random graph state is polynomial in $n$ with very high probability. This is related to the fact that quantum networks with very strong properties arise from random graph states \cite{ascoli2026graphsvertexminoruniversal} and random bipartite graph states \cite{Cautres2024}.

It may be insightful to investigate the geometric measure of entanglement of circle graph states.
For example, the geometric measure of entanglement of circle graph states being too high would yield yet another proof that circle graph states are not universal resources for MBQC.

\subsubsection{Circle graphs in graph theory}

Circle graphs are structurally important in graph theory and feature in a number of open problems.

It is known that every graph of sufficiently large rank-width contains a vertex-minor isomorphic to any given circle graph~\cite{geelenGridTheoremVertexminors2023}.
Similarly, graphs of large rank-width are conjectured to contain pivot-minors isomorphic to any bipartite circle graph (\Cref{conj:bipartite_pivot_minor}; see \cite{Oum2009}).

Oum conjectures that graphs are well-quasi-ordered by the pivot-minor relation (\Cref{conj:wqo_pivot}; see Question 6 in \cite{rankWidthSurvey}) and by the vertex-minor relation (\Cref{conj:wqo_lc}; see \cite{kim2024vertex}).

Geelen conjectures that MBQC is efficiently classically simulable on any vertex-minor-closed class of graphs (\Cref{conj:simulation}; see also \cite{mccarty2021local}, Section 1.4]).

\subsubsection{Quantum communication and quantum cryptography}

Finally, our results are relevant for quantum cryptography protocols that use graph states as resources for communication in quantum networks.

Several communication protocols rely on the distribution of circle graph states for functionalities ranging from (anonymous) quantum conference key agreement using GHZ-states~\cite{murtaQuantumConferenceKey2020, PRXQuantum.1.020325, PRXQuantum.3.040306} or linear cluster states~\cite{deJong2023anonymousconference},
over repeater graph states \cite{takou2025optimization, azuma2015all}
to quantum secret sharing~\cite{hillery1999quantum}.

Since repeater graph states, linear cluster states and GHZ states are all circle graph states, Theorem~\ref{thm:circle_graphs_closed} implies that for all the above communication protocols, the LU-equivalence of the resource states reduces to LC-equivalence, and state convertibility becomes efficiently decidable.

\section{Background and Definitions}

\subsection{Graph theory}

A \textit{graph} $G = (V(G),E(G))$ is composed of two sets, a set $V(G)$ of vertices, and a set $E(G)$ of edges connecting two vertices each, i.e., a subset of $\{\{u,v\} : u,v \in V(G)\}$.

Here, we consider undirected and finite graphs, so edges do not have a preferred direction and $V(G)$ is a finite set. Unless otherwise specified, the graphs are simple, meaning there is at most one edge between two distinct vertices, and no edge connects a vertex to itself.

For some constructions, we make use of \textit{multigraphs}, i.e.~of graphs that may contain loops and multiple edges between the same pair of vertices. A multigraph $M$ consists of a vertex set $V(M)$ and an edge set $E(M)$, where $E(M)$ is a multiset over the set of unordered pairs $\{\{u,v\} : u,v \in V(M)\}$.

We use the notation $u \sim_G v$ when $(u,v) \in E(G)$, and we say $u$ and $v$ are \textit{adjacent} or \textit{neighboring}.

Notice that $u \sim_G v \iff v \sim_G u$ as $G$ is undirected.

Given a vertex $u \in V(G)$, $N_G(u) = \{v\in V(G) ~|~ u \sim_G v\}$ is the \textit{neighborhood} of $u$, i.e., the set of vertices adjacent to $u$.
Two vertices $u$ and $v$ that share the same common neighborhood, $N_v\setminus\{u\}=N_u\setminus\{v\}$, are said to be \textit{twins}.

A graph $G$ is called \textit{bipartite} if its vertex set can be split into two disjoint sets $A$ and $B$ such that every edge of $G$ connects a vertex in $A$ to a vertex in $B$.
In simpler terms, the vertices of a bipartite graph may be colored into two colors such that no edge connects vertices of the same color.
Simple graphs are bipartite if and only if they do not contain a cycle of odd length. (See the violet box of \Cref{fig:circle_graphs_as_vertex_minors_of_bipartite_circle_graphs} for an example of a graph that is not bipartite, and the orange box of the same figure for one that is.)

\subsection{Graph states}

Graph states are quantum states that are in one-to-one correspondence with graphs. Given some graph $G$, the corresponding graph state is constructed by first creating a qubit in the state $\ket + = \frac{1}{\sqrt 2}\left(\ket 0 + \ket 1\right)$ for each vertex, then applying a controlled-$Z$-gate $CZ$ for each edge. The $CZ$-gate is defined as $CZ \ket{x_1}\ket{x_2} = (-1)^{x_1 x_2} \ket{x_1}\ket{x_2}$ where $x_1, x_2 \in \{0,1\}$.

\begin{definition}[Graph state]\label{def:graph_state}
    Let $G$ be a graph of order $n$. The corresponding graph state $\ket G$ is the $n$-qubit state: $$\ket G = \left[\prod_{\{u,v\} \in E(G)} CZ_{uv}\right] \ket{+}_{V(G)}$$
\end{definition}

Graph states are a subfamily of the stabilizer states. An $n$-qubit stabilizer state is the unique fixpoint (i.e.,~the unique +1 eigenvector), up to a global phase, of $n$ independent commuting Pauli operators of the form $\pm \bigotimes_{u\in V} P_u$ where $P_u \in \{I,X,Y,Z\}$.

Note that the stabilizers of a stabilizer state (i.e.,~the Pauli operators which the stabilizer state is a fixpoint of) cannot contain $-I$, because it has no +1 eigenvectors. The stabilizers of a graph state are given by its connectivity, more precisely, for every vertex $u$ of a graph, the corresponding graph state is stabilized by $X_u Z_{N_G(u)}$, i.e.,~$X_u Z_{N_G(u)} \ket G = \ket G$.

To define local Clifford equivalence of graph states, we introduce the Pauli group and the (local) Clifford group.

The \textit{Pauli group} $P_{u}$
is a finite subgroup of order $4^{n+1}$ in the unitary group $\mathcal{U}\left(2^{n}\right)$. On $n$ qubits
$P_{u}$ it is defined as
    \begin{equation}\label{eq:pauli_group_multi_qubit}
        \mathcal{P}_{n}\defeq\left\langle\{\pm 1, \pm i\} \cdot\left\{P_{1} \otimes P_{2} \otimes \cdots \otimes P_{n}\right\}\right\rangle,
    \end{equation}
where the tensor factors $P_{u}$ are arbitrary Pauli matrices.

The normalizer of the Pauli group $\mathcal{P}_{n}$
in $\mathcal{U}\left(2^{n}\right)$
is called the \textit{Clifford group} $\mathcal{C}_{n}$.
As the elements of the normalizer of a subgroup
commute with the entire group as a set,
the Pauli group is invariant under
conjugation with Clifford group elements.
\begin{definition}[Clifford group]\label{def:Clifford_group}
    The $n$ qubit Clifford group $\mathcal{C}_{n}$
    is the group of unitary matrices
    $U \in \mathcal{U}\left(2^{n}\right)$ satisfying
    $U \mathcal{P}_{n} U^{\dagger}$ $=\mathcal{P}_{n}$.
\end{definition}

This allows us to define the
\textit{local Clifford group}, $\mathcal{C}_{n}^{l}$.

\begin{definition}[Local Clifford group]\label{def:local_clifford_group}
    The local Clifford group $\mathcal{C}_{n}^{l}$
    is the subgroup of $\mathcal{C}_{n}$
    that contains all $n$-fold tensor products
    of elements in $\mathcal{C}_{1}$.
\end{definition}

Not only is every graph state a stabilizer state, but graph states are a representative family of stabilizer states. Indeed, every stabilizer state is local Clifford (LC) equivalent to some graph state. LC-equivalence is now properly defined below.

\begin{definition}[LC-equivalence]\label{def:lc_equivalence}
    Two $n$ qubit quantum states $\ket{\psi_1}$ and $\ket{\psi_2}$ are LC-equivalent if there exist $\bigotimes_{u\in V}C_u \in \mathcal{C}^{l}_{n}$ such that $\ket{\psi_2} = e^{i\phi} \bigotimes_{u\in V}C_u \ket{\psi_1}$.
\end{definition}

Quantum states can be equivalent in more general ways than by LC-equivalence.
For example, two quantum states can be equivalent by stochastic local operations and classical communication (SLOCC), by local operations and classical communication (LOCC), or by local unitary operations (LU).
For graph states (and thus stabilizer states), all these notions of local equivalence coincide \cite{Verstraete2003, heinEntanglementGraphStates2006}, thus LU-equivalence, which we define properly below, is the most general notion of local equivalence for graph states.

\begin{definition}[LU-equivalence]\label{def:lu_equivalence}
    Two quantum states $\ket{\psi_1}$ and $\ket{\psi_2}$ are LU-equivalent if there exist single-qubit unitary operators $U_u$ such that $\ket{\psi_2} = \bigotimes_{u\in V}U_u \ket{\psi_1}$.
\end{definition}

We say that $\text{LU} = \text{LC}$ for a stabilizer state $\ket \psi$, if any stabilizer state LU-equivalent to $\ket \psi$, is also LC-equivalent to $\ket \psi$. We say that $\text{LU} = \text{LC}$ for a family of stabilizer states, if $\text{LU} = \text{LC}$  for any element of the family. While it was once conjectured that $\text{LU} = \text{LC}$ for all graph states \cite{PhysRevA.71.062323} (which is the same as saying that $\text{LU} = \text{LC}$ for all stabilizer states), pairs of graph states that are LU-equivalent but not LC-equivalent exist \cite{ji2008lulcconjecturefalse, tsimakuridzeGraphStatesLocal2017}. The smallest known example has 27 qubits, and it is known that $\text{LU} = \text{LC}$ for graph states with 19 qubits or less \cite{claudet2025deciding}.

\subsection{Local complementation and r-local complementation} \label{subsec:rlc}

\begin{figure}[htbp]
  \centering
  \includegraphics[width=0.8\columnwidth]{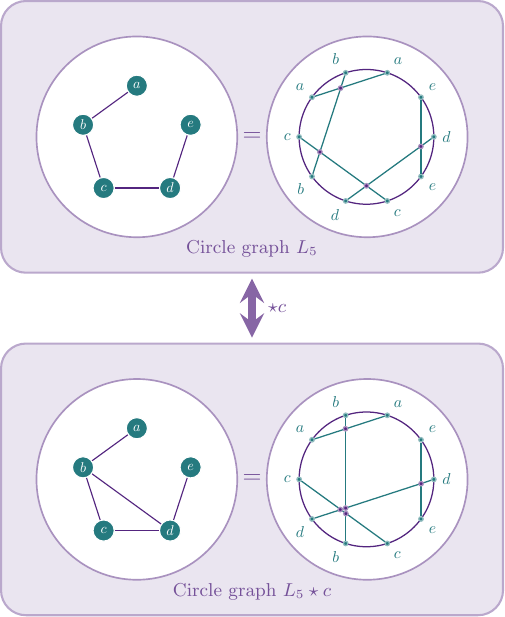}
  \caption{Circle graphs $L_5$ and $L_5 \star c$ and their chord representations.}
  \label{fig:circle_graph_examples}
\end{figure}

Orbits of graph states under reversible local operations can neatly be described in terms of graphical transformations of the underlying graphs.
For LC-operations, the corresponding graph transformations are local complementations.
For arbitrary LU-operations, the corresponding graph transformations are the more general $r$-local complementations.

\textit{Local complementation} with respect to a vertex $u$ of $G$ complements the subgraph induced by the neighborhood of $u$:
Formally, this creates the graph $G\star u= G\Delta K_{N_G(u)}$, where $\Delta$ denotes the symmetric difference and $K_{N_G(u)}$ is the complete graph on the neighborhood of $u$.
(Given two sets $A, B$, the symmetric difference of $A$ and $B$ is $A \Delta B \defeq (A \cup B) \sm (A \cap B)$.)
See \Cref{fig:circle_graph_examples} for an example.

Notably, local complementations on graphs capture all local Clifford transformations of the corresponding graph states \cite{vandennestGraphicalDescriptionAction2004, heinEntanglementGraphStates2006}, i.e.,~two graph states are LC-equivalent if and only if the corresponding graphs are related by a sequence of local complementations.

The alternating sequence of three local complementations with respect to the vertices of an edge $(u,v)$ of a graph $G$, i.e.,
 $G \star  u \star  v \star  u  = G \star  v \star  u \star  v, $
is known as \textit{pivoting} or edge-local complementation.

Local complementations can be generalized to $r$-local complementations \cite{claudet2024local}.
These $r$-local complementations then capture all local unitary transformations of graph states: Two graph states are LU-equivalent if and only if the corresponding graphs are related by a sequence of $r$-local complementations.

To define $r$-local complementation, we first introduce \textit{independent} (multi-) sets and the notion of \textit{$r$-incidence}.

A set of vertices $D \se V(G)$ is said to be \textit{independent} if no two vertices of $D$ are adjacent. In other words, for each $u \in D$, $N_G(u) \cap D = \emptyset$.

With a slight abuse of notation we identify multisets $S \in V(G)^{\mathbb N}$ of vertices with their multiplicity function $S:V(G)\to \mathbb N$. The \textit{support} $\supp(S)$ of a multiset $S$ is then the set of vertices $u$ with $S(u) \gs 1$.
A multiset is said independent if its support is independent.

With this, we are ready to define both \textit{$r$-incidence} and \textit{$r$-local complementation}.

\begin{definition}[$r$-Incidence \cite{claudet2024local}]\label{def:rincidence}
Given a graph $G$, a multiset $S$ of vertices is  $r$-incident, if for any $k\in [0,r)$, and any $K\subseteq V(G)\setminus \supp(S)$ of size $k+2$, $\sum_{u \in \bigcap_{v\in K}N_G(v)}S(u)$ is a multiple of $2^{r-k-\delta(k)}$,
where $\delta$ is the Kronecker delta ($\delta(x)\in \{0,1\}$ and $\delta(x)=1 \Leftrightarrow x=0$) and $\sum_{u \in \bigcap_{v\in K}N_G(v)}S(u)$ is the number of vertices of $S$, counted with their multiplicity, that are neighbors to all vertices of $K$.
\end{definition}

\begin{definition}[$r$-Local Complementation \cite{claudet2024local}]\label{def:rLC}
Given a graph $G$ and an $r$-incident independent multiset $S$, the graph $G\star^r S$ has an edge $a\sim_{G\star^r S} b$ if and only if
\[\left[a\sim_{G} b ~~\oplus~~ \!\!\!\!\!\!\!\sum_{u \in N_G(a) \cap N_G(b)}\!\!\!\!\!\!\! S(u) = 2^{r-1}\bmod 2^{r}\right].\]
\end{definition}

For $r=1$, we recover regular local complementation. Specifically, if $K = \{u_1,u_2, \cdots, u_k\}$ is an independent set, then $G \star^1 K = G \ \star u_1 \star u_2 \star \cdots \star u_k$.

We say that $G\star^r S$ is valid if the multiset $S$ is independent and $r$-incident.
This notion of validity is important, because only valid $r$-local complementations correspond to local unitary transformations in the graph state picture \cite{claudet2024local}. Refer to \cite{claudet2025localequivalencesgraphstates} (Section 4) for an accessible introduction to $r$-local complementation.

\subsection{Vertex-minors and pivot-minors}
In graph theory, local complementation and pivoting are studied in the form of vertex-minors and pivot-minors.
These concepts relate directly to the transformation of quantum graph states under reversible gates and non-reversible measurements.

\begin{definition}[Vertex-minor]
    A graph $H$ is a vertex-minor of a graph $G$ if $H$ can be obtained from $G$ by local complementations and vertex deletions.
\end{definition}

By name, vertex-minors first appeared in \cite{RWandVM}, but the concept was previously known as an $l$-reduction \cite{bouchetCircleGraphObstructions1994}. Note that here we consider vertex-minors of graphs with a fixed labeling and not isomorphism classes of graphs.
The operational interpretation of vertex-minors for quantum graph states is as follows:

\begin{proposition}[see Proposition 5 in \cite{Cautres2024}]
    A graph $H$ is a vertex-minor of a graph $G$ if and only if $\ket{G}$ transforms into $\ket{H}$ via destructive Pauli measurements, local Clifford gates, and classical communication.
\end{proposition}

While this characterization was originally proven strictly for target graphs without isolated vertices (see Theorem 2.6 in \cite{transformingStatesOther}), we adopt the generalized formulation of Ref.~\cite{Cautres2024}, which naturally accommodates isolated vertices through the use of destructive Pauli measurements.

\textit{Destructive} measurements emphasize that the measured qubit is no longer considered part of the system.
Here, it is useful to refer to destructive measurements in order to obtain an equivalence between these two notions.
For example, the 2-vertex empty graph $\overline{K_2}$ is not a vertex-minor of the 2-vertex complete graph $K_2$; however, the graph state corresponding to $\overline{K_2}$ can be obtained from the graph state corresponding to $K_2$ with non-destructive $Z$-measurements.

The pivot-minor relation is analogous to the vertex-minor relation, but local complementation is replaced by pivoting.

\begin{definition}[Pivot-minor]
    A graph $H$ is a pivot-minor of a graph $G$ if $H$ can be obtained from $G$ by pivotings and vertex deletions.
\end{definition}

The pivot-minor relation is stronger than the vertex-minor relation, in the sense that if $H$ is a pivot-minor of $G$, then $H$ is also vertex-minor of $G$. The converse is however not true (for example, the 3-cycle is a vertex-minor but not a pivot-minor of the 4-cycle). Similar to vertex-minors, pivot-minors capture local Clifford operations and Pauli measurements on quantum graph states, but with a more restricted set of operations.

\begin{proposition}
    A graph $H$ is a pivot-minor of a graph $G$ if and only if $\ket{G}$ transforms into $\ket{H}$ via destructive $X$- or $Z$-measurements, Hadamard gates, Pauli gates, and classical communication.
\end{proposition}

\begin{proof}
    ($\rightarrow$) Pivoting can be implemented with Hadamard and $Z$-gates \cite{mhalla2012graph}, and vertex deletion can be implemented with a $Z$-measurement, along with possible Pauli corrections depending on the outcome of the measurement \cite{hein2004multiparty}.

    ($\leftarrow$) $Z$-measurement corresponds to vertex-deletion, up to some possible Pauli corrections. $X$-measurement on a (non-isolated) qubit $a$ corresponds to choosing a vertex $b$ adjacent to $a$, applying a pivoting on the edge $ab$, then deleting $a$, up to some possible Pauli corrections and a Hadamard gate on $b$ \cite{hein2004multiparty}. After the measurements, the graph state obtained is equivalent to $H$ up to Hadamard gates and Pauli gates. Thus, these two graphs a related by pivotings (see for example proof of Proposition 6 in \cite{claudet2024local}).
\end{proof}

Both vertex-minors and pivot-minors feature in important structural conjectures by Oum in graph theory.

\begin{conjecture}[Oum; Question 6 in \cite{rankWidthSurvey}] \label{conj:wqo_pivot}
    Graphs are well-quasi-ordered by the pivot-minor relation.
\end{conjecture}
In other words the conjecture asks if every infinite sequence of graphs $G_1, G_2,\ldots$ contains a pair $G_i, G_j$ (with $i < j$) such that $G_i$ is isomomorphic to a pivot-minor of $G_j$.
For bipartite graphs, the conjecture is already resolved: Bipartite graphs are well-quasi-ordered by the pivot-minor relation~\cite{geelen2014solving}.

A weaker conjecture, implied by the previous one, is asking for the well-quasi-ordering by the vertex-minor relation instead:

\begin{conjecture}[\cite{kim2024vertex}] \label{conj:wqo_lc}
    Graphs are well-quasi-ordered by the vertex-minor relation.
\end{conjecture}

A family of graphs known to be well-quasi-ordered by the vertex-minor relation is the family of circle graphs (implied by the well-quasi-ordering of 4-regular graphs via an immersion relation of Robertson and Seymour \cite{robertson2010graph, kim2024vertex}).

\subsection{Circle graphs}\label{subsec:circle_graphs}

\begin{figure*}[htbp]
\centering
  \includegraphics[width=\textwidth]{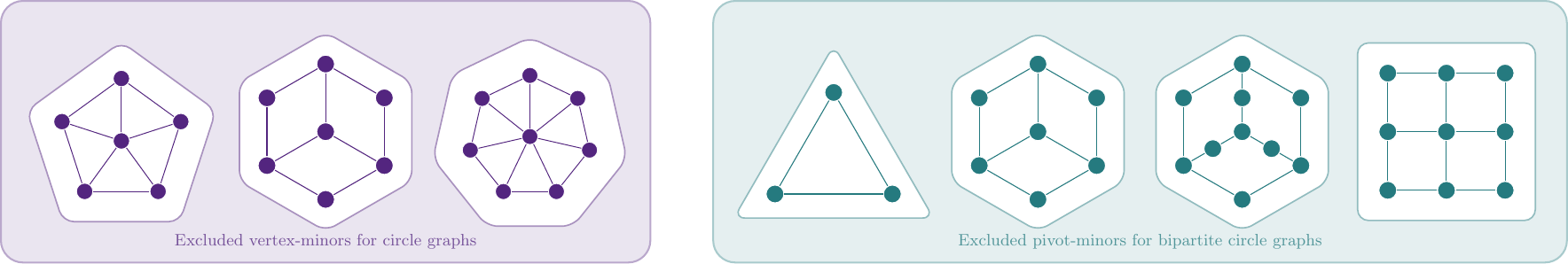}
\caption{
Circle graphs and bipartite circle graphs can both be characterized by excluded minors:
\textcolor{violet}{A simple graph is a circle graph if and only if it does not have a vertex-minor isomorphic to any of the three graphs depicted in the violet box \cite{bouchetCircleGraphObstructions1994}.}
\textcolor{green}{A simple graph is a bipartite circle graph if and only if it does not have a pivot-minor isomorphic to any of the four graphs depicted in the green box.}
Note that all obstructions that characterize all circle graphs (bipartite or not) via excluded pivot-minors can be found in~\cite{geelen2009circle}.
}
\label{fig:excluded_minors}
\end{figure*}

Circle graphs are defined as follows, and examples are illustrated in \Cref{fig:circle_graph_examples}.

\begin{definition}[Circle graph]\label{def:circle_graph}
A graph $G$ is a circle graph if it is isomorphic to the intersection graph of a finite number of chords of a circle.
\end{definition}

Circle graphs are closed under local complementation and vertex deletion, in other words, they are closed by the vertex-minor relation.
They can be elegantly represented in terms of forbidden vertex-minors.

\begin{proposition}
    A graph is a circle graph if and only if none of its vertex-minors is isomorphic to a graph in \Cref{fig:excluded_minors}, left.
\end{proposition}

Graph classes characterized by a family of forbidden vertex-minors (such as circle graphs) are precisely the vertex-minor-closed graph classes.
Such classes feature in the Geelen simulation conjecture.

\begin{conjecture}[Geelen, see \cite{mccarty2021local}, Section 1.4]\label{conj:simulation}
    MBQC is efficiently classically simulable on any vertex-minor-closed class of graphs, excluding the class of all graphs.
\end{conjecture}

Two examples of vertex-minor-closed graph classes for which this conjecture is true are given by circle graphs \cite{harrison2025fermionicinsightsmeasurementbasedquantum}, and graphs of rank-width at most $k$, where $k$ is an arbitrary integer \cite{van2006universal}.

In this manuscript, we also consider circle graphs that are bipartite.

Bipartite circle graphs have chord diagrams whose chords are divided into two sets of chords that only intersect with chords from the other set and never with chords from the same set (see, e.g.~\Cref{fig:planar_code_states_are_bipartite_circle_graphs} on the right or \Cref{fig:circle_graphs_as_vertex_minors_of_bipartite_circle_graphs} on the bottom).
Bipartite circle graphs are closed by pivoting and vertex deletion, in other words they are closed by the pivot-minor relation.
Also, similar to circle graphs that are characterized by their excluded vertex-minors, bipartite circle graphs can be  characterized by their excluded pivot-minors.

\begin{proposition}
    A graph is a bipartite circle graph if and only if none of its pivot-minors is isomorphic to a graph in \Cref{fig:excluded_minors}, right.
\end{proposition}

\begin{proof}
    A graph is bipartite if any only if none of its pivot-minors is isomorphic to a triangle.
    (The triangle is a pivot-minor of every larger odd cycle. Conversely, a graph is bipartite if and only if it contains no odd cycle. Every pivot-minor of a bipartite graph is bipartite, since neither pivots nor vertex deletions create an odd cycle if none existed before.)

    Geleen and Oum showed that a graph is a circle graph if and only if none of its pivot-minors is isomorphic to a graph in a set of fifteen graphs~\cite{geelen2009circle}.
    Bipartite circle graphs are the intersection of bipartite graphs and circle graphs, thus their forbidden pivot-minors are those fifteen graphs together with the triangle.
    Twelve graphs in this fifteen-graph set are not bipartite, meaning they contain a triangle as a vertex-minor. Therefore, these graphs can be disregarded.
\end{proof}

Bipartite circle graphs feature in an other important conjecture by Oum that link rank-width to pivot-minors.

\begin{conjecture}[\cite{Oum2009}]\label{conj:bipartite_pivot_minor}
    For every bipartite circle graph $H$, every graph $G$ of sufficiently large rank-width contains a pivot-minor isomorphic to $H$.
\end{conjecture}

\begin{figure*}[htbp]
\centering
  \includegraphics[width=\textwidth]{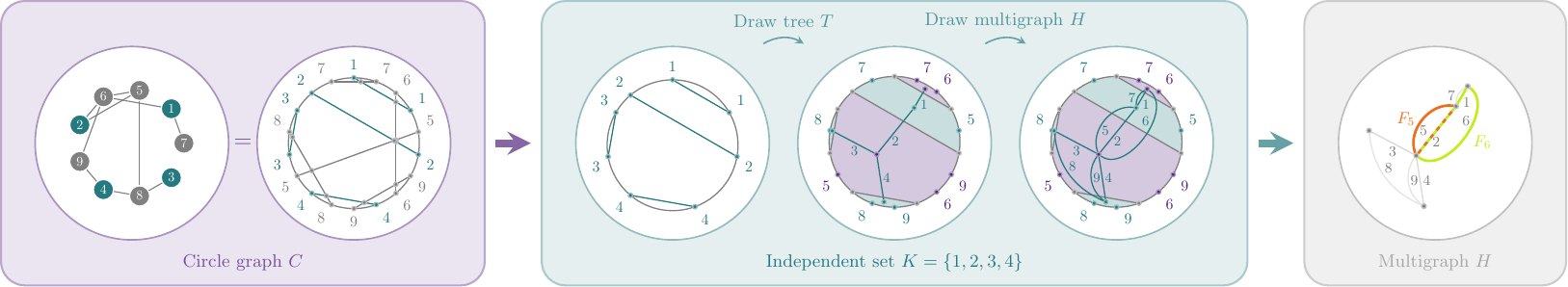}
\caption{
    For an exemplary \textcolor{violet}{$9$-vertex circle graph $C$} with a choice of \textcolor{green}{independent set $K=\{1,2,3,4\}$}, we show how to construct the corresponding tree $T$ and \textcolor{gray}{multigraph $H$}.
    \textbf{\textcolor{violet!70}{Violet arrow:}} From the chord representation of $C$, we
    remove all chords not corresponding to $K=\{1,2,3,4\}$.
    This yields a simpler diagram of four
    nonintersecting chords that divide the circle into five regions.
    In each of the five circle regions, we then draw
    a vertex and connect those pairs of vertices
    that correspond to pairs of neighboring regions
    by a tree edge. We label each of the four tree edges with the
    name of the chord that separates
    the neighboring circle regions.
    \textbf{\textcolor{green!70}{Green arrow:}}
    Adding all vertices $e\in V(G) \setminus K$
    (that lie outside the independent set $K$)
    as additional edges to this tree $T$,
    creates fundamental cycles $F_e$ in a new \textcolor{gray}{multigraph $H$}.
    As an example, the fundamental cycles $\textcolor{orange}{F_5}$ and $\textcolor{brightgreen}{F_6}$ are highlighted.
    Note that in the chord diagram of $C$, $e$ intersects
    a different chord $v$ if and only if $v$ is an edge
    of the fundamental cycle $F_e$ of $H$. In our example, $F_5$, containing edge $2$, corresponds to the $5$ and the $2$ chords intersecting the chord diagram of $C$ (and $F_6$, containing edges $1, 2$, corresponds to $6$ intersecting both $1, 2$ in the chord diagram of $C$).
}
\label{fig:multigraph_H_cycles}
\end{figure*}

To study the effect of $r$-local complementation on circle graphs, it is crucial to understand the structure of independent sets within circle graphs.
(Compare to Fig.~\ref{fig:multigraph_H_cycles} for an example.)

    For every independent set $K$
    of a given circle graph $C$,
    we can construct a tree $T$
    whose edges correspond to the vertices
    of this independent set~\cite{Fraysseix1981}.
    This tree, $T$, can then in turn be used to construct a multigraph, $H$, that is essential for our proof of Lemma~\ref{lemma:circle_no_twins} and therefore of Theorem~\ref{thm:circle_graphs_closed}.
    The construction works as follows.

    From a chord representation of $C$, we
    remove all chords not corresponding to
    our independent set $K$.
    This gives us a chord diagram of $k=|K|$
    nonintersecting chords of a circle that divide the circle into $k+1$ regions.

    We now draw the tree $T$:
    For each of the $k+1$ circle regions we draw
    a vertex and connect each pair of vertices
    corresponding to a pair of neighboring regions
    by a tree edge. We label each tree edge by the
    name of the chord that separates
    the two circle regions in the chord diagram of $K$.
    Notice that this construction results in a tree
    precisely because no two such chords intersect.

    Note further that adding any single vertex $e\in V(G) \setminus K$
    (that lies outside our independent set $K$)
    as an additional edge to this tree $T$,
    creates a single fundamental cycle in a new graph $H$ with edges $E(H)=E(T)\cup \{e\}$.
    We denote this fundamental cycle by $F_e$.
    In the chord diagram of $C$, $e$ intersects
    a different chord $v$ if and only if $v$ is an edge
    of the fundamental cycle $F_e$ of $H$.

The above construction for independent sets within circle graphs will be instrumental in proving the technical \Cref{lemma:circle_no_twins} of \Cref{sec:LU=LC_for_circle_graphs}.

\subsection{Planar code states}\label{subsec:planar_code_states}

Planar code states are a special class of stabilizer states. More precisely, they are Calderbank--Shor--Steane (CSS) states \cite{Calderbank96, Steane96}, meaning that there exists a basis for their stabilizer group such that each stabilizer contains only X, or only Z.
Like all stabilizer states, CSS states are LC-equivalent to some graph states. In particular, CSS states are known to be exactly the stabilizer states LC-equivalent to bipartite, i.e. 2-colorable, graph states \cite{chen2004multi}.

The qubits of a planar code state correspond to the edges of a planar multigraph. We allow two or more edges between two vertices (hence the denomination ``multigraph''). We also allow self-loops (i.e. edges with both ends being the same vertex), but we may ignore them as they correspond to isolated qubits.
Given a planar multigraph $P$ with $n$ edges, the corresponding $n$-qubit planar code state is defined by the following generators of its $2^n$ element stabilizer group:

\begin{itemize}
    \item $\bigotimes_{e\in \partial f} Z_e$ for every face $f$ of $P$ (including the exterior),  where $\partial f$ denotes the set of edges forming the boundary of the face $f$,
\item$\bigotimes_{e\ni v} X_e$
for every vertex $v$ of $P$, where $e\ni v$ means that the edge $e$ is incident to the vertex $v$.
\end{itemize}

Note that according to Euler's theorem for planar graphs ($|V|-|E|+|F|=2$, where $|V|$, $|E|$ and $|F|$ are the numbers of vertices, edges, and faces, respectively), the number of stabilizer generators described as above is equal to two plus the number of qubits, i.e.,~$n+2$. Therefore, two of these generators are redundant.
For an illustration of a planar code state as it is typically presented, refer to \Cref{fig:planar_code_states_are_bipartite_circle_graphs}.

\section{Circle graphs are closed under r-local complementation}\label{sec:LU=LC_for_circle_graphs}

In this section we prove that the only graph states LU-equivalent to circle graph states are circle graph states themselves. To this end, we prove that circle graph states are closed by $r$-local complementation (see \Cref{subsec:rlc}).

Informally, this is done by proving that independent $r$-incident sets, required to apply $r$-local complementation, only occur in circle graphs in a trivial form. First, we prove the following technical lemma, restricting the independent $r$-incident sets we have to consider.

\begin{lemma} \label{lemma:notwins}
    If $S$ is an independent $r$-incident multiset of a graph $G$, then there exists an independent $r$-incident multiset $S'$ such that:
    \begin{itemize}
        \item $\supp(S')$ contains no twins;
        \item $\supp(S')$ contains no vertex of degree 0 or 1;
        \item for every vertex $u$, $S'(u) \in [0, 2^{r-1}-1]$;
    \end{itemize}
     and an independent set $A$ of vertices, such that $G \star^r S = G \star^1 A \star^r S' $.
\end{lemma}

\begin{proof}
    First, we can construct a set $S_0$ from $S$ as follows.
    For every vertex $u$ of degree 0 or 1 in $\supp(S)$, let $S_0(u) = 0$. For every set of twins $u_1, u_2, \cdots, u_k$ in $\supp(S)$ (that are not of degree 1), let $S_0(u_1) = S(u_1) + S(u_2) + \cdots + S(u_k)$ and $S_0(u_i)=0$ for $i \in [2,k]$. For every other vertex $u$, let $S_0(u) = S(u)$.

    Second, let $A$ be the set of vertices such that
    $(S_0(u) \bmod 2^{r})\gs 2^{r-1}$.
    For every vertex $u$, let $S'(u) = S_0(u) \bmod 2^{r-1}$.

    It is straightforward to check that $S'$ is $r$-incident. Also, $G \star^r S = G \star^{r} 2^{r-1} A \star^{r} S' = G \star^1 A \star^r S'$ (where $2^{r-1} A$ is the multiset that contains each vertex of $A$, $2^{r-1}$ times).
\end{proof}

The structure of circle graphs leads to a restriction on their independent sets.

\begin{lemma} \label{lemma:circle_no_twins}
    Given a circle graph C and an independent set K (with at least one vertex of degree two or higher but without twins), there are distinct vertices $a,b \in V(C)\sm K$ with exactly one common neighbor in $K$, i.e., $|N_C(a) \cap N_C(b) \cap K| = 1$.
\end{lemma}

\begin{proof}
    Suppose by contradiction that there exist no two distinct vertices $a,b \in V(C)\sm K$ such that $|N_C(a) \cap N_C(b) \cap K| = 1$. By hypothesis there exists at least one vertex in $K$ of degree two or higher, then there exist distinct vertices $a,b \in V(C)\sm K$ such that $|N_C(a) \cap N_C(b) \cap K| \gs 2$. Suppose without loss of generality that $|N_C(a) \cap N_C(b) \cap K|$ is minimal, i.e.,

\begin{align*}
&|N_C(a) \cap N_C(b) \cap K|\\
  &= \min_{\substack{
        u,v \in V(C)\setminus K,\, u \neq v,\\
        N_C(u)\cap N_C(v)\cap K \neq \emptyset}}
     \left|\,N_C(u) \cap N_C(v) \cap K\,\right|.
\end{align*}

    For an example of the construction that now follows, see \Cref{fig:multigraph_H_cycles}.

    Let $T=(V(T),E(T))$ be the tree corresponding to the independent set $K$. Recall that the vertices in $K$ correspond to edges in $T$. For every vertex $u \in V(C)\sm K$ such that $|N_C(u)\cap K|\gs 1$, we create an edge $e \in \binom{V(T)}{2}$ between two vertices in $V(T)$, such that the single fundamental cycle in $(V(T),E(T)+e)$ contains only the edges in $E(T)$ corresponding to all neighbors of $u$ in $K$. At the end of the procedure, we end up with a \textit{multigraph} $H$ (i.e.,~there may be multiple edges between two vertices), where $|E(H)|=|V(C)|$.

    By construction, the fundamental cycles $F_a$ and $F_b$ corresponding to $a$ and $b$, respectively, share $|N_C(a) \cap N_C(b) \cap K| \gs 2$ edges in $E(H)$. Let $u,v \in |N_C(a) \cap N_C(b) \cap K|$ be two of the common neighbors of $a$ and $b$ in $K$.

    Hence, $F_a$ and $F_b$ in $H$ overlap on the edges $u$ and $v$.

    By hypothesis, $u$ and $v$ are not twins, thus there exists a vertex $x \in V(C) \sm K$
    that is a neighbor to only one of $u$ or $v$.
    Note that if $F_x$ contains both an edge in $E(F_a)\setminus E(F_b)$ and an edge in $E(F_b)\setminus E(F_a)$, then both $u$ and $v$ would be edges in $F_x$. Thus, without loss of generality, by symmetry between $a$ and $b$, we may assume that $F_x$ does not contain any edge in $E(F_b)\setminus E(F_a)$. Then $N_C(x) \cap N_C(b) \cap K$ is smaller than $N_C(a) \cap N_C(b) \cap K$, contradicting that $|N_C(a) \cap N_C(b) \cap K|$ is minimal. (Note that it is strictly smaller because we lose at least one of $u,v$.)
\end{proof}

We now apply both \Cref{lemma:notwins}
and \Cref{lemma:circle_no_twins}.

\begin{theorem} \label{thm:circle_graphs_closed}
For circle graph states every valid $r$-local complementation is implemented by some sequence of local complementations.
    Equivalently, circle graph states satisfy $\text{LU} = \text{LC}$ and circle graphs are closed under $r$-local complementation.
\end{theorem}

\begin{proof}
    Let $C$ be a circle graph, and $S$ be an independent $r$-incident multiset of $C$. According to \Cref{lemma:notwins}, there exists an independent $r$-incident multiset $S'$ such that
    \begin{itemize}
        \item $\supp(S')$ contains no twins;
        \item $\supp(S')$ contains no vertex of degree 0 or 1;
        \item for every vertex $u$, $S'(u) \in [0, 2^{r-1}-1]$;
    \end{itemize}
    and a set $A$ of vertices, such that $C \star^r S = C \star^1 A \star^r S' $.
    We now show that $S'$ is empty.

    Suppose by contradiction that $\supp(S')$ is non-empty.
    Thanks to our application of \Cref{lemma:notwins},  we know that then $\supp(S')$ contains at least one vertex of degree two or higher, and that $\supp(S')$ does not contain twins.

    Thus we can apply \Cref{lemma:circle_no_twins}:
    There exist two distinct vertices $a, b \in V(C)\sm \supp(S')$ such that $|N_C(a) \cap N_C(b) \cap \supp(S')| = 1$, implying that $S'$ is not $r$-incident. Indeed, if $w\in\supp(S')$ is the unique common neighbor of $a$ and $b$ in $\supp(S')$, i.e.,~$N_C(a) \cap N_C(b) \cap \supp(S')=\{w\}$ then $\sum_{u \in \bigcap_{v\in \{a,b\}}N_C(v)}S'(u) = S'(w)$, which is not a multiple of $2^{r-1}$, as $S'(w) \in [1,2^{r-1}-1]$.
    (c.f. $k=0, \delta(k)=1$ in the definition of $r$-incidence.)

    Thus, $\supp(S')$ is empty, and $C \star^r S = C \star^1 A$. As circle graphs are closed under local complementation, $C \star^r S$ is a circle graph.
\end{proof}

In the subsequent section, we will see that this has immediate consequences for MBQC
on circle graph states by leveraging a one-to-one correspondence between planar code states and bipartite circle graph states.

\section{Planar code states are bipartite circle graph states} \label{sec:planar_code_states}

The planar code states introduced in \Cref{subsec:planar_code_states} correspond to bipartite circle graph states  introduced in \Cref{subsec:circle_graphs}.

\begin{figure*}[htbp]
\centering
  \includegraphics[width=\textwidth]{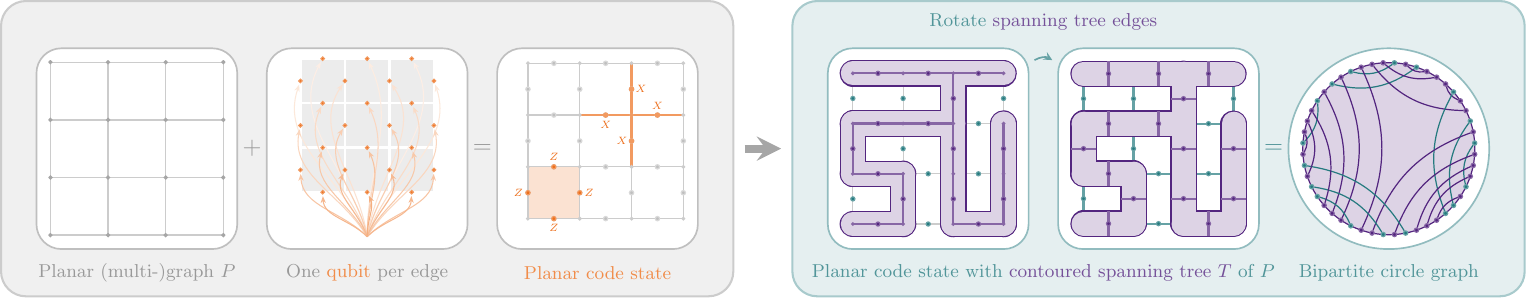}
\caption{Illustration of Theorem~\ref{thm:planar_code_states_are_bipartite_circle_graph_states}.
The gray box shows an example of a \textcolor{orange}{planar code state $\ket{\psi}$}, where the corresponding \textcolor{gray}{planar multigraph $P$} is the $(4\times 4)$-grid graph without any multideges.
The green box then shows how an arbitrary selection of a \textcolor{violet}{spanning tree $T$} of $P$ allows for the construction of a corresponding bipartite circle graph by applying a Hadamard gate to \textcolor{green}{each qubit outside of the spanning tree}.
}
\label{fig:planar_code_states_are_bipartite_circle_graphs}
\end{figure*}

\begin{theorem}\label{thm:planar_code_states_are_bipartite_circle_graph_states}
  Every planar code state is LC-equivalent to a bipartite circle graph state. Conversely, every bipartite circle graph state is LC-equivalent to a planar code state.
\end{theorem}

\begin{proof}
  Let $P = (V(P),E(P))$ be a planar multigraph, and $\ket{\psi}$ be the corresponding planar code state.
  The qubits of $\ket{\psi}$ correspond to the edges of $P$
  (cf.~gray box of Figure \ref{fig:planar_code_states_are_bipartite_circle_graphs}).

  Let $T$ be a spanning tree of $P$. Notice $V(T)=V(P)$ and $E(T) \se E(P)$. Let $\ket{C} \defeq H_{u\in E(P)\sm E(T)} \ket{\psi}$ be the state obtained from $\ket\psi$ by applying a Hadamard gate on the qubits corresponding to the edges in $E(P)$ not contained in $E(T)$. We will show that $\ket C$ is a bipartite circle graph state
  (cf.~green box of Figure \ref{fig:planar_code_states_are_bipartite_circle_graphs}).

  Since $P$ is planar, each edge in $a \in E(P)\sm E(T)$ induces a unique fundamental cycle $F_a \se E(P)$ when added to its spanning tree $T$.

  First, let us show that for each edge $a \in E(P)\sm E(T)$, $X_a Z_{F_a\sm\{a\}}$ stabilizes $\ket C$, i.e.,~$X_a Z_{F_a\sm\{a\}} \ket C = \ket C$. The product of the stabilizers of $\ket \psi$ corresponding to the faces enclosed by $F_a$ is $Z_{F_a}$. Indeed, every edge in $E(P)\sm E(T)$ inside the area delimited by $F_a$, appears exactly twice in the product. Then,
  \begin{align*}
    X_a Z_{F_a\sm\{a\}} \ket C & = X_a Z_{F_a\sm\{a\}} H_{u\in E(P)\sm E(T)} \ket{\psi}\\
    & = H_{u\in E(P)\sm E(T)} Z_{F_a} \ket{\psi}\\
    & = H_{u\in E(P)\sm E(T)} \ket{\psi} = \ket C
  \end{align*}

  Second, let us show that that for each edge $t \in E(T)$, $X_t Z_{F^{-1}_t}$ stabilizes $\ket C$, where $F^{-1}_t \se E(P) \sm E(T)$ is the set of edges $a \in E(P) \sm E(T)$ such that $t \in F_a$.
  Edge $t$ separates $T$ into two unique trees $T_1$ and $T_2$, such that $V(T) = V(T_1) \sqcup V(T_2)$ and $E(T) = \{t\} \sqcup E(T_1) \sqcup E(T_2)$. The product of the stabilizers of $\ket \psi$ corresponding to the vertices in $V(T_1)$ is $X_{F^{-1}_t \cup \{t\}}$.

  Indeed, no edge in $E(T_2)$ appears in the product, and each edge in $E(T_1)$ appears twice in the product, while $t$ appears exactly once in the product. Also, no edge in $E(P) \sm E(T)$ whose fundamental cycle lies in $E(T_2)$ appears in the product, and each edge in $E(P) \sm E(T)$ whose fundamental cycle lies in $E(T_1)$ appears twice in the product. The edges that appear once in the product are exactly those edges in $E(P) \sm E(T)$ with one end in $V(T_1)$ and the other end in $V(T_2)$, i.e.,~those edges whose fundamental cycle contains $t$.

  Then,
  \begin{align*}
    X_t Z_{F^{-1}_t} \ket C & = X_t Z_{F^{-1}_t} H_{u\in E(P)\sm E(T)}\ket{\psi}\\
    & = H_{u\in E(P)\sm E(T)} X_{F^{-1}_t \cup \{t\}} \ket{\psi}\\
    & = H_{u\in E(P)\sm E(T)} \ket{\psi} = \ket C
  \end{align*}

  Thus, $\ket C$ is a bipartite graph state, i.e.,~$C=(L(C)\sqcup R(C),E(C))$ where both $L(C)$ and $R(C)$ are independent sets. The leftmost vertices in $L(C)$ correspond to edges in $E(T)$, and rightmost vertices in $R(C)$ correspond to edges in $E(P) \sm E(T)$. The connectivity of $C$ follows from the fundamental cycles of $T$. More precisely, a rightmost vertex $r \in R(C)$ is connected to a leftmost vertex $l \in L(C)$ if and only if the edge in $E(T)$ corresponding to $l$ is in the fundamental cycle induced by the edge in $E(P)\sm E(T)$ corresponding to $r$.

  In graph-theoretical terms, the bipartite graph $C$ is a fundamental graph of the planar multigraph $P$. It follows that $C$ is a bipartite \textbf{circle} graph, as bipartite circle graphs are exactly the fundamental graphs of planar multigraphs \cite{Fraysseix1981}.

  Conversely, for any bipartite circle graph $C$ there is a planar multigraph $P$ such that $C$ is a fundamental graph of $P$.
\end{proof}

\begin{remark}
    In general, there are many possible spanning trees of a planar multigraph. Choosing a different spanning tree may lead to a different bipartite circle graph, but those different bipartite circle graphs are related by pivotings.
    This makes sense as pivoting on bipartite graphs are implemented with Hadamard gates \cite{mhalla2012graph}.
    Choosing a spanning tree amounts to choosing where to apply these Hadamard gates.
\end{remark}

As we showed that for circle graph states, LC- and LU-equivalence coincide, the same goes for planar code states.

\begin{corollary}\label{cor:planar_code_states_LU_LC}
  Planar code states satisfy $\text{LU} = \text{LC}$.
\end{corollary}

This was proven in \cite{Sarvepalli2010PRA}, but with some non-trivial restrictions, that translate to the corresponding bipartite circle graph not having any vertex of degree 1 or twins. For example, their proof does not work for the planar code state defined of a 2D grid (see \Cref{fig:planar_code_states_are_bipartite_circle_graphs} for example).

Planar code states are not efficient universal resources for MBQC.

\begin{theorem}[\cite{Bravyi2007,Bravyi2022}]\label{thm:planar_code_states_ECS}
  MBQC is efficiently classically simulable on the planar code state.
\end{theorem}

MBQC being efficiently classically simulable means here that there exists a classical polynomial time algorithm that approximately samples from the same output distribution as any sequence of (possibly adaptative) single-qubit measurements. \Cref{thm:planar_code_states_ECS} was first proved in \cite{Bravyi2007} under the assumption that the measured qubits are connected, but this restriction was removed in \cite{Bravyi2022}. This result then translates to bipartite circle graphs, as MBQC being efficiently classically simulable for a class of quantum states is invariant by local operations.

\begin{corollary}\label{cor:MBQC_ECS_on_bipartite_circle_graphs}
  MBQC is efficiently classically simulable on bipartite circle graphs.
\end{corollary}

To lift \Cref{cor:MBQC_ECS_on_bipartite_circle_graphs} to (possibly non-bipartite) circle graphs, we first need to reduce circle graphs to bipartite circle graphs.

\begin{proposition}\label{prop:vertex_minor_of_bipartite_circle_graph}
  Every $n$-vertex circle graph is a vertex-minor of some $2n^2$-vertex bipartite circle graph.
\end{proposition}

\begin{proof}
A graph is a circle graph if and only if it is the vertex-minor of a bipartite circle graph (see e.g.~Corollary 53 of \cite{brijderIsotropicMatroidsII2016}).
We will now show an upper bound on the number of vertices of this bipartite circle graph.
Our construction is illustrated by an example in Figure~\ref{fig:circle_graphs_as_vertex_minors_of_bipartite_circle_graphs}.

For any given $n$-vertex circle graph $C$, we start with a corresponding chord diagram.
From this chord diagram, we can read off (clockwise or counterclockwise) a double occurrence word and draw a 4-regular multigraph with vertices corresponding to the word's letters and edges for all adjacent letters, where the first and the last letter are also considered adjacent.
 The drawing of this 4-regular multigraph may have a crossing number greater than zero, i.e.~a number $f(n)>0$ points in which its edges cross; let us call these points its failures of planarity.

 Now, replace every failure of planarity with an additional vertex. We obtain an $n+f(n)$ vertex, 4-regular multigraph that is planar. According to Theorem 50 in \cite{brijderIsotropicMatroidsII2016}, a graph is the interlacement graph of a planar 4-regular graph if and only if it is locally equivalent to a bipartite graph. Thereby, we have shown that there exists a bipartite circle graph $B$, of which our original circle graph $C$ is a vertex-minor. (Note that the graph state $\ket{C}$ can be obtained from $\ket{B}$ by measuring the $f(n)$ additional qubits in the $Y$-basis. See Theorem 2.6~in Ref.~\cite{HowTransformGrapha})

 What is left to show is that $f(n) \leqslant 2n^2-n$ is an upper bound for the crossing number.
 Note that any 4-regular multigraph has exactly $2n$ edges.
 In a good drawing, none of these edges crosses itself and no two edges cross each other more than once.
 A crossing therefore always corresponds to an unordered pair of edges.
 Since the number of unordered pairs of edges is given by $$\binom{2n}{2}= \frac{(2n)!}{(2n-2)!2!}= \frac{(2n)(2n-1)}{2}=2n^2-n,$$ we have found an upper bound on $f(n)$.
\end{proof}

\begin{figure}[htbp]
\centering
  \includegraphics[width=\columnwidth]{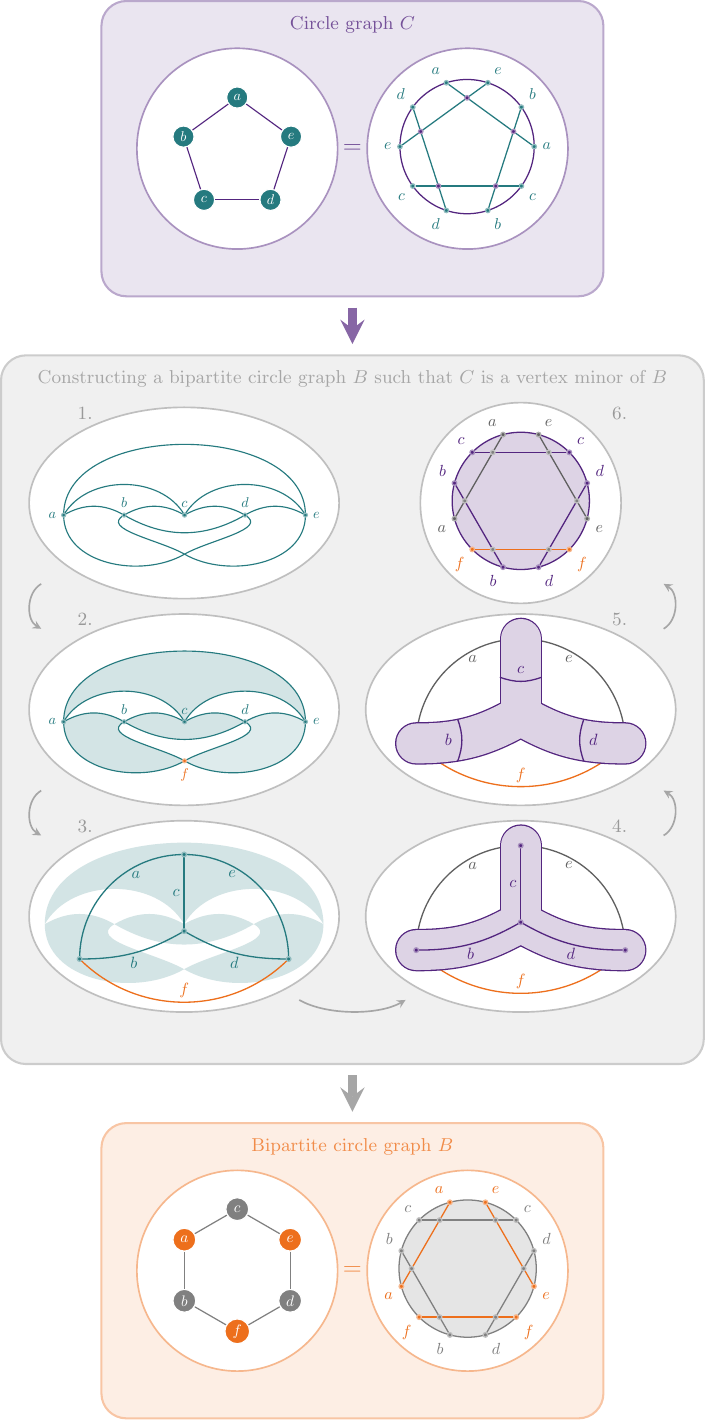}
\caption{
Every circle graph $C$ is a vertex-minor of some bipartite circle graph $B$.
We construct such a $B$ for the example where $C$ is a 5-cycle:
$1.$~From a chord diagram of $C$, we read off a double occurrence word ($aebacbdced$) and \textcolor{green}{draw a 4-regular multigraph} with vertices $a,b,c,d,e$ and edges for all adjacent letters in this word, i.e.,~we contract the chords.
$2.$~This 4-regular multigraph may have nonzero crossing number, but we can planarize it. Here, the crossing number is one and thus \textcolor{orange}{adding a single vertex $f$} yields a planar 4-regular multigraph. We then 2-color the faces of this graph in green and white.
$3.$~Choose one face color (w.l.o.g.~green).
Draw a new graph with vertices for each green face and edges for each point where the green faces touch.
$4.$~\textcolor{violet}{Select a spanning tree and draw a contour around it.}
$5.$~Rotate the edges of the spanning tree so that they connect to the contour.
$6.$~Transform the contour into a circle, yielding a chord diagram of a bipartite circle graph $B$.
}
\label{fig:circle_graphs_as_vertex_minors_of_bipartite_circle_graphs}
\end{figure}

With this, we can show that MBQC is efficiently classically simulable on circle graph states.

Let $\ket{C_n}$ be any $n$-qubit given circle graph state.
By \Cref{prop:vertex_minor_of_bipartite_circle_graph}, we know that there exists a $m$-qubit bipartite circle graph $B_m$ with $m=\mathrm{poly}(n)= 2n^2$ such that $C_n$ can be obtained from $B_m$ by a sequence of operations that are themselves efficiently classically simulable (here Pauli measurements and local Clifford corrections based on the classical feedforward of the measurement outcomes), and that this sequence can be computed in classical polynomial time (see steps in \Cref{fig:circle_graphs_as_vertex_minors_of_bipartite_circle_graphs}).

Now to simulate an MBQC on $\ket{C_n}$, we instead simulate the equivalent procedure on $\ket{B_m}$: we include the classical description of the reduction from $B_m$ to $C_n$ into the measurement pattern, run the known polynomial-time classical simulator for $B_m$ on input size $m$, and classically track the additional local Clifford corrections which require only polynomial time.

Since the above construction runs in time $\mathrm{poly}(m)$ and $m=\mathrm{poly}(n)$, and the overhead from preprocessing and postprocessing is also polynomial, the total runtime stays polynomial in $n$.
With this, we recover the main result of \cite{harrison2025fermionicinsightsmeasurementbasedquantum}.

\begin{corollary}[\cite{harrison2025fermionicinsightsmeasurementbasedquantum}]\label{cor:MBQC_ECS_on_circle_graph_states}
  MBQC is efficiently classically simulable on circle graph states.
\end{corollary}

Since circle graphs have polynomial rank-width (see Appendix \ref{app:rank_width_of_circle_graphs}), we then recover the fact that while polynomial rank-width is a necessary condition for efficient universality, it is not a sufficient one~\cite{Bravyi2007}.

\begin{corollary}[\cite{Bravyi2007}]\label{cor:rank-width_sufficient_condition}
Assuming $\BPP \neq \BQP$, polynomial rank-width is not a sufficient condition for efficient universality of a graph state resource for MBQC.
\end{corollary}

\section{Counting LU-equivalent graph states is \#P-hard}

The problem of deciding whether two graph states are LC-equivalent belongs to \NP~by verifying a sequence of local complementations. However, because deciding the LC-equivalence of the underlying graphs is in \Poly~\cite{Bouchet1991}, the graph state problem actually falls into the stricter class of \Poly~\cite{VdnEfficientLC}.

By definition, this means that the problem of counting the number of graph states LC-equivalent to a given graph state is in \sharpP. Dahlberg, Helsen and Wehner proved that the problem is actually \sharpP-complete, i.e.,~it is one of the hardest problems in \sharpP \cite{dahlberg2020counting}. This was done by showing that the \sharpP-complete problem of counting the number of Eulerian tours of a 4-regular multigraph, can be reduced in polynomial time to computing the number of graph states LC-equivalent to a given circle graph state.

\begin{theorem}[\cite{dahlberg2020counting}]
    Counting the number of graph states LC-equivalent to a given graph state is \sharpP-complete, even when restricted to circle graph states.
\end{theorem}

According to \Cref{thm:circle_graphs_closed}, for circle graph states, LC and LU-equivalence coincide, so the number of graph states LU-equivalent to a given circle graph state, is the same as the number of graph states LC-equivalent to it. Thus, counting the number of graph states LU-equivalent to a given graph state is \sharpP-hard.

\begin{corollary}\label{cor:counting}
    Counting the number of graph states LU-equivalent to a given graph state is \sharpP-hard, even when restricted to circle graph states.
\end{corollary}

It is not known if the problem of deciding whether two graph states are LU-equivalent is in \NP, i.e.,~if the problem of counting the number of graph states LU-equivalent to a given graph state is in \sharpP. It is only known that the LU-equivalence of graph states can be decided in quasi-polynomial time, i.e.,~in time complexity $n^{O(\log(n))}$ where $n$ is the number of qubits \cite{claudet2025deciding}. Proving that the LU-equivalence of graph states is in \NP, would prove that counting the number of graph states LU-equivalent to a given graph state is \sharpP-complete.

\begin{acknowledgments}
The authors would like to thank Jim Geelen for many helpful discussions on circle graphs.
Further, the authors would like to express their gratitude to Anna Pappa, Isaac D. Smith, Robert Rau\ss{}endorf and Hans J. Briegel for useful discussions.
This research was funded in part by the Austrian Science Fund (FWF) [SFB BeyondC F7102, DOI: 10.55776/F71; WIT9503323, DOI: 10.55776/WIT9503323]. For open access purposes, the authors have applied a CC BY public copyright license to any author-accepted manuscript version arising from this submission.
This work was also supported by the European Union via the ERC Advanced Grant, QuantAI, (No. 101055129) and via the Quantum Internet Alliance (No. 101102140). The views and opinions expressed in this article are however those of the author(s) only and do not necessarily reflect those of the European Union or the European Research Council -- neither the European Union nor the granting authority can be held responsible for them.
\end{acknowledgments}

\bibliographystyle{quantum}
\bibliography{references}

@inproceedings{claudet2024local,
	archiveprefix = {arXiv},
	author = {Nathan Claudet and Simon Perdrix},
	booktitle = {Proceedings of the 42nd {S}ymposium on {T}heoretical {A}spects of  {C}omputer Science ({STACS} 2025)},
	eprint = {2409.20183},
	primaryclass = {quant-ph},
	title = {Local equivalence of stabilizer states: a graphical characterisation},
	year = {2025},
	doi ={10.4230/LIPIcs.STACS.2025.27}
}

@inproceedings{claudet2025deciding,
	archiveprefix = {arXiv},
	author = {Nathan Claudet and Simon Perdrix},
	booktitle = {Proceedings of the 52nd International Colloquium on Automata, Languages, and Programming (ICALP 2025)},
	eprint = {2502.06566},
	primaryclass = {quant-ph},
	title = {Deciding Local Unitary Equivalence of Graph States in Quasi-Polynomial Time},
	year = {2025},
	doi ={10.4230/LIPIcs.ICALP.2025.59}
}

@article{vandennestGraphicalDescriptionAction2004,
	title = {Graphical description of the action of local {Clifford} transformations on graph states},
	volume = {69},
	issn = {1050-2947, 1094-1622},	
	doi = {10.1103/PhysRevA.69.022316},
	language = {english},
	number = {2},
	urldate = {2022-06-24},
	journal = {Physical Review A},
	author = {Van den Nest, Maarten and Dehaene, Jeroen and De Moor, Bart},
	month = feb,
	year = {2004},
	pages = {022316},
    eprint={quant-ph/0308151}
}

@article{HowTransformGrapha,
	author = {Axel Dahlberg and Jonas Helsen and Stephanie Wehner},
	doi = {10.1088/2058-9565/aba763},
	eprint = {1805.05306},
	journal = {Quantum Science and Technology},
	month = {Sep},
	number = {4},
	pages = {045016},
	publisher = {IOP Publishing},
	title = {How to transform graph states using single-qubit operations: computational complexity and algorithms},
	volume = {5},
	year = {2020}
}

@article{transformingStatesOther,
	author = {Dahlberg, A.  and Wehner, S.},
	title = {Transforming graph states using single-qubit operations},
	journal = {Philosophical Transactions of the Royal Society A: Mathematical, Physical and Engineering Sciences},
	volume = {376},
	number = {2123},
	pages = {20170325},
	year = {2018},
	doi = {10.1098/rsta.2017.0325},	
	eprint = {1805.05305}
}

@article{jelinekRankwidthSquareGrid2010,
	title = {The rank-width of the square grid},
	volume = {158},
	issn = {0166218X},
	doi = {10.1016/j.dam.2009.02.007},
	language = {en},
	number = {7},
	urldate = {2023-10-26},
	journal = {Discrete Applied Mathematics},
	author = {Jelínek, Vít},
	month = apr,
	year = {2010},
	pages = {841--850},
}

@article{PRXQuantum.1.020325,
  title = {Anonymous Quantum Conference Key Agreement},
  author = {Hahn, Frederik and de Jong, Jarn and Pappa, Anna},
  journal = {PRX Quantum},
  volume = {1},
  issue = {2},
  pages = {020325},
  numpages = {12},
  year = {2020},
  month = {Dec},
  publisher = {American Physical Society},
  doi = {10.1103/PRXQuantum.1.020325},  
  eprint={2003.10186}
}

@article{murtaQuantumConferenceKey2020,
	title = {Quantum {Conference} {Key} {Agreement}: {A} {Review}},
	volume = {3},
	issn = {2511-9044},
	shorttitle = {Quantum {Conference} {Key} {Agreement}},	
	doi = {10.1002/qute.202000025},
	language = {en},
	number = {11},
	urldate = {2022-05-09},
	journal = {Advanced Quantum Technologies},
	author = {Murta, Gláucia and Grasselli, Federico and Kampermann, Hermann and Bruß, Dagmar},
	year = {2020},
	note = {\_eprint: https://onlinelibrary.wiley.com/doi/pdf/10.1002/qute.202000025},
	pages = {2000025},
    eprint={2003.10186}
}

@article{dahlbergComplexityVertexminorProblem2022,
	title = {The complexity of the vertex-minor problem},
	volume = {175},
	issn = {0020-0190},	
	doi = {10.1016/j.ipl.2021.106222},
	language = {en},
	urldate = {2022-05-15},
	journal = {Information Processing Letters},
	author = {Dahlberg, Axel and Helsen, Jonas and Wehner, Stephanie},
	month = apr,
	year = {2022},
	pages = {106222},
    eprint={1906.05689}
}

@misc{brijderIsotropicMatroidsII2016,
	title = {Isotropic matroids {II}: {Circle} graphs},
	shorttitle = {Isotropic matroids {II}},
	doi = {10.37236/5223},
	language = {en},
	publisher = {arXiv},
	author = {Brijder, Robert and Traldi, Lorenzo},
	month = oct,
	year = {2016},
    eprint={1504.04299},
    journal={The Electronic journal of Combinatorics},
    issue={Volume 23, Issue 4}
}

@article{geelenGridTheoremVertexminors2023,
	title = {The {Grid} {Theorem} for vertex-minors},
	volume = {158},
	issn = {00958956},	
	doi = {10.1016/j.jctb.2020.08.004},
	language = {english},
	urldate = {2023-10-10},
	journal = {Journal of Combinatorial Theory, Series B},
	author = {Geelen, Jim and Kwon, O-joung and McCarty, Rose and Wollan, Paul},
	month = jan,
	year = {2023},
	pages = {93--116},
    eprint={1909.08113}
}

@phdthesis{mccarty2021local,
  title={Local structure for vertex-minors},
  author={McCarty, Rose},
  school={University of Waterloo},
  type={PhD thesis},
  address={Waterloo, Ontario, Canada},
  year={2021},
  url={https://uwspace.uwaterloo.ca/items/1cfbfc52-2e30-44a4-b3fb-28493c3d94f0}  
}

@article{heinEntanglementGraphStates2006,
	author = {Hein, Marc and D{\"u}r, Wolfgang and Eisert, Jens and Raussendorf, Robert and Maarten Van den Nest and Briegel, Hans J.},
	doi = {10.3254/978-1-61499-018-5-115},
	eprint = {quant-ph/0602096},
	journal = {Quantum computers, algorithms and chaos},
	month = {Mar},
	title = {Entanglement in Graph States and its Applications},
	volume = {162},
	year = {2006}
}

@article{Sarvepalli2010PRA,
  title = {Local equivalence, surface-code states, and matroids},
  author = {Sarvepalli, Pradeep and Raussendorf, Robert},
  journal = {Phys. Rev. A},
  volume = {82},
  issue = {2},
  pages = {022304},
  numpages = {12},
  year = {2010},
  month = {Aug},
  publisher = {American Physical Society},
  doi = {10.1103/PhysRevA.82.022304},  
  eprint={0911.1571}
}

@article{Bravyi2007,
  title = {Measurement-based quantum computation with the toric code states},
  author = {Bravyi, Sergey and Raussendorf, Robert},
  journal = {Phys. Rev. A},
  volume = {76},
  issue = {2},
  pages = {022304},
  numpages = {10},
  year = {2007},
  month = {Aug},
  publisher = {American Physical Society},
  doi = {10.1103/PhysRevA.76.022304},  
  eprint={quant-ph/0610162}
}

@article{vandennestClassicalSimulationUniversality2007,
	title = {Classical simulation versus universality in measurement-based quantum computation},
	volume = {75},
	copyright = {http://link.aps.org/licenses/aps-default-license},
	issn = {1050-2947, 1094-1622},	
	doi = {10.1103/PhysRevA.75.012337},
	language = {en},
	number = {1},
	urldate = {2025-07-28},
	journal = {Physical Review A},
	author = {Van Den Nest, M. and Dür, W. and Vidal, G. and Briegel, H. J.},
	month = jan,
	year = {2007},
	pages = {012337},
    eprint={quant-ph/0608060}
}

@article{Bravyi2022,
  title = {How to Simulate Quantum Measurement without Computing Marginals},
  author = {Bravyi, Sergey and Gosset, David and Liu, Yinchen},
  journal = {Phys. Rev. Lett.},
  volume = {128},
  issue = {22},
  pages = {220503},
  numpages = {6},
  year = {2022},
  month = {Jun},
  publisher = {American Physical Society},
  doi = {10.1103/PhysRevLett.128.220503},
  eprint={2112.08499}
}

@article{ji2008lulcconjecturefalse,
author = {Ji, Zhengfeng and Chen, Jianxin and Wei, Zhaohui and Ying, Mingsheng},
title = {The {LU-LC} conjecture is false},
year = {2010},
issue_date = {Jan 2010},
publisher = {Rinton Press, Incorporated},
address = {Paramus, NJ},
volume = {10},
number = {1},
issn = {1533-7146},
journal = {Quantum Information and Computation},
month = {Jan},
pages = {97–108},
numpages = {12},
eprint={0709.1266},
doi = {QIC10.1-2-8.html}
}

@article{PhysRevA.71.062323,
  title = {Local unitary versus local {C}lifford equivalence of stabilizer states},
  author = {Van den Nest, Maarten and Dehaene, Jeroen and De Moor, Bart},
  journal = {Phys. Rev. A},
  volume = {71},
  issue = {6},
  pages = {062323},
  numpages = {7},
  year = {2005},
  month = {Jun},
  publisher = {American Physical Society},
  doi = {10.1103/PhysRevA.71.062323},
  eprint={quant-ph/0411115}
}

@article{tsimakuridzeGraphStatesLocal2017,
	author = {Nikoloz Tsimakuridze and Otfried G{\"u}hne},
	doi = {10.1088/1751-8121/aa67cd},
	journal = {Journal of Physics A: Mathematical and Theoretical},
	month = {Apr},
	number = {19},
	pages = {195302},
	publisher = {{IOP} Publishing},
	title = {Graph states and local unitary transformations beyond local {C}lifford operations},
	volume = {50},
	year = 2017,
	eprint={1611.06938}
}

@article{PRXQuantum.3.040306,
  title = {Secure Anonymous Conferencing in Quantum Networks},
  author = {Grasselli, Federico and Murta, Gl\'aucia and de Jong, Jarn and Hahn, Frederik and Bru\ss{}, Dagmar and Kampermann, Hermann and Pappa, Anna},
  journal = {PRX Quantum},
  volume = {3},
  issue = {4},
  pages = {040306},
  numpages = {23},
  year = {2022},
  month = {Oct},
  publisher = {American Physical Society},
  doi = {10.1103/PRXQuantum.3.040306},
  eprint={2111.05363}
}

@article{hillery1999quantum,
  title={Quantum secret sharing},
  author={Hillery, Mark and Bu{\v{z}}ek, Vladim{\'\i}r and Berthiaume, Andr{\'e}},
  journal={Physical Review A},
  volume={59},
  number={3},
  pages={1829},
  year={1999},
  publisher={APS},
  doi={10.1103/PhysRevA.59.1829},
  eprint={quant-ph/9806063}
}

@article{deJong2023anonymousconference,
  doi = {10.22331/q-2023-09-21-1117},  
  title = {Anonymous conference key agreement in linear quantum networks},
  author = {de Jong, Jarn and Hahn, Frederik and Eisert, Jens and Walk, Nathan and Pappa, Anna},
  journal = {{Quantum}},
  issn = {2521-327X},
  publisher = {{Verein zur F{\"{o}}rderung des Open Access Publizierens in den Quantenwissenschaften}},
  volume = {7},
  pages = {1117},
  month = sep,
  year = {2023},
  eprint={2205.09169}
}

@article{Goff2012,
  title = {Classical simulation of measurement-based quantum computation on higher-genus surface-code states},
  author = {Goff, Leonard and Raussendorf, Robert},
  journal = {Phys. Rev. A},
  volume = {86},
  issue = {4},
  pages = {042301},
  numpages = {18},
  year = {2012},
  month = {Oct},
  publisher = {American Physical Society},
  doi = {10.1103/PhysRevA.86.042301},  
  eprint={1201.6319}
}

@misc{harrison2025fermionicinsightsmeasurementbasedquantum,
      title={Fermionic Insights into Measurement-Based Quantum Computation: Circle Graph States Are Not Universal Resources},
      author={Brent Harrison and Vishnu Iyer and Ojas Parekh and Kevin Thompson and Andrew Zhao},
      year={2025},
      doi = {10.48550/arXiv.2510.05557},
      eprint={2510.05557},
      archivePrefix={arXiv},
      primaryClass={quant-ph}
}

@article{geelen2009circle,
  title={Circle graph obstructions under pivoting},
  author={Geelen, Jim and Oum, Sang-il},
  journal={Journal of Graph Theory},
  volume={61},
  number={1},
  pages={1--11},
  year={2009},
  publisher={Wiley Online Library},
  doi={10.1002/jgt.20363},
}

@article{bouchetCircleGraphObstructions1994,
	title = {Circle {Graph} {Obstructions}},
	volume = {60},
	issn = {0095-8956},
	doi = {10.1006/jctb.1994.1008},
	language = {en},
	number = {1},
	urldate = {2022-05-13},
	journal = {Journal of Combinatorial Theory, Series B},
	author = {Bouchet, A.},
	month = jan,
	year = {1994},
	pages = {107--144},
}

@article{Fraysseix1981,
title = {Local complementation and interlacement graphs},
journal = {Discrete Mathematics},
volume = {33},
number = {1},
pages = {29-35},
year = {1981},
issn = {0012-365X},
doi = {10.1016/0012-365X(81)90255-7},
author = {Hubert {de Fraysseix}},
}

@article{dahlberg2020counting,
  title={Counting single-qubit {C}lifford equivalent graph states is {\#}{P}-complete},
  author={Dahlberg, Axel and Helsen, Jonas and Wehner, Stephanie},
  journal={Journal of Mathematical Physics},
  volume={61},
  number={2},
  year={2020},
  publisher={AIP Publishing},
  doi={10.1063/1.5120591},
  eprint={1907.08024}
}

@article{Gross2009,
  title = {Most Quantum States Are Too Entangled To Be Useful As Computational Resources},
  author = {Gross, D. and Flammia, S. T. and Eisert, J.},
  journal = {Phys. Rev. Lett.},
  volume = {102},
  issue = {19},
  pages = {190501},
  numpages = {4},
  year = {2009},
  month = {May},
  publisher = {American Physical Society},
  doi = {10.1103/PhysRevLett.102.190501},  
  eprint={0810.4331}
}

@article{Ghosh2025,
  title = {Random Regular Graph States Are Complex at Almost Any Depth},
  author = {Ghosh, Soumik and Hangleiter, Dominik and Helsen, Jonas},
  journal = {PRX Quantum},
  volume = {6},
  issue = {4},
  pages = {040344},
  numpages = {40},
  year = {2025},
  month = {Nov},
  publisher = {American Physical Society},
  doi = {10.1103/52xz-3hpc},  
  eprint={2412.07058}
}

@inproceedings{Cautres2024,
	author = {Cautr\`{e}s, Maxime and Claudet, Nathan and Mhalla, Mehdi and Perdrix, Simon and Savin, Valentin and Thomass\'{e}, St\'{e}phan},
	title =	{Vertex-Minor Universal Graphs for Generating Entangled Quantum Subsystems},
	booktitle = {Proceedings of the 51st International Colloquium on Automata, Languages, and Programming (ICALP 2024)},
	doi = {10.4230/LIPIcs.ICALP.2024.36},
	year = {2024},
	eprint={2402.06260}
}

@article{VandenNest2007,
doi = {10.1088/1367-2630/9/6/204},
year = {2007},
month = {jun},
publisher = {},
volume = {9},
number = {6},
pages = {204},
author = {Van den Nest, M and Dür, W and Miyake, A and Briegel, H J},
title = {Fundamentals of universality in one-way quantum computation},
journal = {New Journal of Physics},
eprint={quant-ph/0702116}
}

@article{shi2006classical,
  title={Classical simulation of quantum many-body systems with a tree tensor network},
  author={Shi, Y-Y and Duan, L-M and Vidal, Guifre},
  journal={Physical Review A—Atomic, Molecular, and Optical Physics},
  volume={74},
  number={2},
  pages={022320},
  year={2006},
  publisher={APS},
  doi={10.1103/PhysRevA.74.022320},
  eprint={quant-ph/0511070}
}

@article{van2006universal,
  title={Universal resources for measurement-based quantum computation},
  author={Van den Nest, Maarten and Miyake, Akimasa and D{\"u}r, Wolfgang and Briegel, Hans J},
  journal={Physical review letters},
  volume={97},
  number={15},
  pages={150504},
  year={2006},
  publisher={APS}, 
  doi = {10.1103/PhysRevLett.97.150504},
  eprint={quant-ph/0604010}
}

@article{takou2025optimization,
  title={Optimization complexity and resource minimization of emitter-based photonic graph state generation protocols},
  author={Takou, Evangelia and Barnes, Edwin and Economou, Sophia E},
  journal={npj Quantum Information},
  volume={11},
  number={1},
  pages={108},
  year={2025},  
  publisher={Nature Publishing Group UK London},
  eprint={2407.15777},
  doi={10.1038/s41534-025-01056-3}
}

@article{azuma2015all,
  title={All-photonic quantum repeaters},
  author={Azuma, Koji and Tamaki, Kiyoshi and Lo, Hoi-Kwong},
  journal={Nature communications},
  volume={6},
  number={1},
  pages={6787},
  year={2015},  
  doi={10.1038/ncomms7787},
  publisher={Nature Publishing Group UK London},
  eprint={1309.7207}
}

@article{Liao2021,
  title = {Graph-state representation of the toric code},
  author = {Liao, Pengcheng and Feder, David L.},
  journal = {Phys. Rev. A},
  volume = {104},
  issue = {1},
  pages = {012432},
  numpages = {16},
  year = {2021},
  month = {Jul},
  publisher = {American Physical Society},
  doi = {10.1103/PhysRevA.104.012432},  
  eprint={2103.12268}
}

@article{Calderbank96,
  title = {Good quantum error-correcting codes exist},
  author = {Calderbank, A. R. and Shor, Peter W.},
  journal = {Phys. Rev. A},
  volume = {54},
  issue = {2},
  pages = {1098--1105},
  numpages = {0},
  year = {1996},
  month = {Aug},
  publisher = {American Physical Society},
  doi = {10.1103/PhysRevA.54.1098},  
  eprint={quant-ph/9512032}
}

@article{Steane96,
  title = {Error Correcting Codes in Quantum Theory},
  author = {Steane, A. M.},
  journal = {Phys. Rev. Lett.},
  volume = {77},
  issue = {5},
  pages = {793--797},
  numpages = {0},
  year = {1996},
  month = {Jul},
  publisher = {American Physical Society},
  doi = {10.1103/PhysRevLett.77.793}  
}

@article{chen2004multi,
  doi     = {10.26421/QIC7.8-1},
  title   = {Multi-partite quantum cryptographic protocols with noisy {GHZ} states},
  author  = {Chen, Kai and Lo, Hoi-Kwong},
  journal = {Quantum Information and Computation},
  volume  = {7},
  number  = {8},
  month   = {Nov},
  year    = {2007},
  eprint  = {quant-ph/0404133}
}

@misc{bauer2026quadratictensorsunificationclifford,
      title={Quadratic tensors as a unification of {C}lifford, {G}aussian, and free-fermion physics},
      author={Andreas Bauer and Seth Lloyd},
      year={2026},
      doi={10.48550/arXiv.2601.15396},
      eprint={2601.15396},
      archivePrefix={arXiv},
      primaryClass={quant-ph}
}

@article {rankWidthSurvey,
    AUTHOR = {Oum, Sang-il},
     TITLE = {Rank-width: algorithmic and structural results},
   JOURNAL = {Discrete Appl. Math.},
  FJOURNAL = {Discrete Applied Mathematics. The Journal of Combinatorial
              Algorithms, Informatics and Computational Sciences},
    VOLUME = {231},
      YEAR = {2017},
     PAGES = {15--24},
      ISSN = {0166-218X,1872-6771},
   MRCLASS = {05C85 (05C05 05C07 05C50 05C69)},
  MRNUMBER = {3695267},
MRREVIEWER = {David\ Burns},
       DOI = {10.1016/j.dam.2016.08.006},       
       eprint={1601.03800}
}

@article{graphMinors10,
    title = {Graph minors. {X}. {O}bstructions to tree-decomposition},
    journal = {Journal of Combinatorial Theory, Series B},
    volume = {52},
    number = {2},
    pages = {153-190},
    year = {1991},
    issn = {0095-8956},
    doi = {https://doi.org/10.1016/0095-8956(91)90061-N},    
    author = {Neil Robertson and P.D Seymour}
}

@article{robertson2010graph,
  title={Graph minors {XXIII}. {N}ash-{W}illiams' immersion conjecture.},
  author={Robertson, Neil and Seymour, Paul D},
  journal={J. Comb. Theory B},
  volume={100},
  number={2},
  pages={181--205},
  year={2010},
  url={https://web.math.princeton.edu/~pds/papers/GM23/GM23.pdf},
  doi={https://doi.org/10.1016/j.jctb.2009.07.003},
}

@article{branchwidthCycleMatroid,
    title = {The branchwidth of graphs and their cycle matroids},
    journal = {Journal of Combinatorial Theory, Series B},
    volume = {97},
    number = {5},
    pages = {681-692},
    year = {2007},
    issn = {0095-8956},
    doi = {https://doi.org/10.1016/j.jctb.2006.12.007},
    author = {Illya V. Hicks and Nolan B. McMurray}
}

@article{RWandVM,
    title = {Rank-width and vertex-minors},
    journal = {Journal of Combinatorial Theory, Series B},
    volume = {95},
    number = {1},
    pages = {79-100},
    year = {2005},
    issn = {0095-8956},
    doi = {10.1016/j.jctb.2005.03.003},    
    author = {Sang-il Oum}
}

@incollection {connectivityIsotropic,
    AUTHOR = {A. Bouchet{}},
     TITLE = {Connectivity of isotropic systems},
 BOOKTITLE = {Combinatorial {M}athematics: {P}roceedings of the {T}hird
              {I}nternational {C}onference ({N}ew {Y}ork, 1985)},
    VOLUME = {555},
     PAGES = {81--93},
 PUBLISHER = {New York Acad. Sci., New York},
      YEAR = {1989},
       DOI = {10.1111/j.1749-6632.1989.tb22439.x}       
}

@article{bramblesTreeCut,
    title = {On objects dual to tree-cut decompositions},
    journal = {Journal of Combinatorial Theory, Series B},
    volume = {157},
    pages = {401-428},
    year = {2022},
    issn = {0095-8956},
    doi = {https://doi.org/10.1016/j.jctb.2022.07.009},    
    author = {Łukasz Bożyk and Oscar Defrain and Karolina Okrasa and Michał Pilipczuk},
    eprint={2103.14667}
}

@article{bramblesSubmodular,
    title = {Submodular partition functions},
    journal = {Discrete Mathematics},
    volume = {309},
    number = {20},
    pages = {6000-6008},
    year = {2009},
    issn = {0012-365X},
    doi = {https://doi.org/10.1016/j.disc.2009.04.033},    
    author = {Omid Amini and Frédéric Mazoit and Nicolas Nisse and Stéphan Thomassé}
}

@inproceedings{certifyingLargeBW,
    author = {Oum, Sang-il and Seymour, Paul},
    title = {Certifying large branch-width},
    year = {2006},
    isbn = {0898716055},
    publisher = {Society for Industrial and Applied Mathematics},
    address = {USA},
    booktitle = {Proceedings of the Seventeenth Annual ACM-SIAM Symposium on Discrete Algorithm},
    pages = {810–813},
    numpages = {4},
    location = {Miami, Florida},
    doi = {10.1145/1109557.1109646},    
    series = {SODA '06}
}

@misc{ascoli2026graphsvertexminoruniversal,
      title={Almost all graphs are vertex-minor universal},
      author={Ruben Ascoli and Bryce Frederickson and Sarah Frederickson and Caleb McFarland and Logan Post},
      year={2026},
      doi = {10.48550/arXiv.2602.09049},
      eprint={2602.09049},
      archivePrefix={arXiv},
      primaryClass={quant-ph}
}

@article{geelen2014solving,
  title={Solving Rota’s conjecture},
  author={Geelen, Jim and Gerards, Bert and Whittle, Geoff},
  journal={Notices of the AMS},
  volume={61},
  number={7},
  pages={736--743},
  year={2014},  
  doi={http://dx.doi.org/10.1090/noti1139}
}

@article{lee2012rank,
  title={Rank-width of random graphs},
  author={Lee, Choongbum and Lee, Joonkyung and Oum, Sang-il},
  journal={Journal of Graph Theory},
  volume={70},
  number={3},
  pages={339--347},
  year={2012},
  publisher={Wiley Online Library},
  eprint={1001.0461},
  doi={10.1002/jgt.20620}
}

@article{mhalla2012graph,
	author = {Mehdi Mhalla and Simon Perdrix},
	journal = {International Journal of Unconventional Computing},
	number = {1-2},
	pages = {153--171},
	title = {Graph States, Pivot Minor, and Universality of ({X},{Z})-Measurements},
	volume = {9},
	year = {2013},
	doi = {10.48550/arXiv.1202.6551},
	eprint = {1202.6551},
	archivePrefix = {arXiv},
	primaryClass = {quant-ph}
}

@article{hein2004multiparty,
  title={Multiparty entanglement in graph states},
  author={Hein, Marc and Eisert, Jens and Briegel, Hans J},
  journal={Physical Review A—Atomic, Molecular, and Optical Physics},
  volume={69},
  number={6},
  pages={062311},
  year={2004},
  publisher={APS}, 
  doi={10.1103/PhysRevA.69.062311},
  eprint={quant-ph/0307130}
}

@article{jozsa2003role,
  title={On the role of entanglement in quantum-computational speed-up},
  author={Jozsa, Richard and Linden, Noah},
  journal={Proceedings of the Royal Society of London. Series A: Mathematical, Physical and Engineering Sciences},
  volume={459},
  number={2036},
  pages={2011--2032},
  year={2003},
  publisher={The Royal Society},
  doi={10.1098/rspa.2002.1097},
  eprint={quant-ph/0201143}
}

@article{kim2024vertex,
  title={Vertex-minors of graphs: A survey},
  author={Kim, Donggyu and Oum, Sang-il},
  journal={Discrete Applied Mathematics},
  volume={351},
  pages={54--73},
  year={2024},
  publisher={Elsevier},   
  doi={10.1016/j.dam.2024.03.011}
}

@article{Raussendorf2001,
  title = {A One-Way Quantum Computer},
  author = {Raussendorf, Robert and Briegel, Hans J.},
  journal = {Phys. Rev. Lett.},
  volume = {86},
  issue = {22},
  pages = {5188--5191},
  numpages = {0},
  year = {2001},
  month = {May},
  publisher = {American Physical Society},
  doi = {10.1103/PhysRevLett.86.5188},
  eprint = {quant-ph/0010033},
  archivePrefix = {arXiv},
  primaryClass = {quant-ph}
}

@phdthesis{claudet2025localequivalencesgraphstates,
      title={Local Equivalences of Graph States},
      author={Nathan Claudet},
      year={2025},
      school={Universit\'e de Lorraine},
      address={Nancy, France},
      doi={10.48550/arXiv.2511.22271},
      eprint={2511.22271},
      archivePrefix={arXiv},
      primaryClass={quant-ph}
}

@article{Oum2009,
author = {Oum, Sang-il},
title = {Excluding a bipartite circle graph from line graphs},
journal = {Journal of Graph Theory},
volume = {60},
number = {3},
pages = {183-203},
doi = {10.1002/jgt.20353},
year = {2009}
}

@misc{bravyi1998quantumcodeslatticeboundary,
      title={Quantum codes on a lattice with boundary},
      author={S. B. Bravyi and A. Yu. Kitaev},
      year={1998},
      doi={10.48550/arXiv.quant-ph/9811052},
      eprint={quant-ph/9811052},
      archivePrefix={arXiv},
      primaryClass={quant-ph}
}

@article{Kitaev2003,
title = {Fault-tolerant quantum computation by anyons},
journal = {Annals of Physics},
volume = {303},
number = {1},
pages = {2-30},
year = {2003},
issn = {0003-4916},
doi = {https://doi.org/10.1016/S0003-4916(02)00018-0},
author = {A.Yu. Kitaev},
eprint={quant-ph/9707021}
}

@article{Bouchet1991,
	author = {André Bouchet},
	day = {01},
	doi = {10.1007/BF01275668},
	issn = {1439-6912},
	journal = {Combinatorica},
	month = {Dec},
	number = {4},
	pages = {315-329},
	title = {An efficient algorithm to recognize locally equivalent graphs},	
	volume = {11},
	year = {1991}
}

@article{VdnEfficientLC,
	author = {Maarten Van den Nest and Dehaene, Jeroen and De Moor, Bart},
	doi = {10.1103/PhysRevA.70.034302},
	issue = {3},
	journal = {Physical Review A},
	month = {Sep},
	numpages = {3},
	pages = {034302},
	publisher = {American Physical Society},
	title = {Efficient algorithm to recognize the local {C}lifford equivalence of graph states},
	volume = {70},
	year = {2004},
	eprint={quant-ph/0405023}
}

@article{Verstraete2003,
  title = {Normal forms and entanglement measures for multipartite quantum states},
  author = {Verstraete, Frank and Dehaene, Jeroen and De Moor, Bart},
  journal = {Physical Review A},
  volume = {68},
  issue = {1},
  pages = {012103},
  numpages = {7},
  year = {2003},
  month = {Jul},
  publisher = {American Physical Society},
  doi = {10.1103/PhysRevA.68.012103},
  eprint={quant-ph/0105090}
}

@article{briegel2001,
  title = {Persistent Entanglement in Arrays of Interacting Particles},
  author = {Briegel, Hans J. and Raussendorf, Robert},
  journal = {Phys. Rev. Lett.},
  volume = {86},
  issue = {5},
  pages = {910--913},
  numpages = {0},
  year = {2001},
  month = {Jan},
  publisher = {American Physical Society},
  doi = {10.1103/PhysRevLett.86.910},
  eprint={quant-ph/0004051}
}

@article{raussendorf2003measurementbased,
  title = {Measurement-based quantum computation on cluster states},
  author = {Raussendorf, Robert and Browne, Daniel E. and Briegel, Hans J.},
  journal = {Phys. Rev. A},
  volume = {68},
  issue = {2},
  pages = {022312},
  numpages = {32},
  year = {2003},
  month = {Aug},
  publisher = {American Physical Society},
  doi = {10.1103/PhysRevA.68.022312},
  eprint={quant-ph/0301052}
}

@article{raussendorf2019computationally,
  title = {Computationally Universal Phase of Quantum Matter},
  author = {Raussendorf, Robert and Okay, Cihan and Wang, Dong-Sheng and Stephen, David T. and Nautrup, Hendrik Poulsen},
  journal = {Phys. Rev. Lett.},
  volume = {122},
  issue = {9},
  pages = {090501},
  numpages = {5},
  year = {2019},
  month = {Mar},
  publisher = {American Physical Society},
  doi = {10.1103/PhysRevLett.122.090501},
  eprint={1803.00095}
}

@article{stephen2019subsystem,
  doi = {10.22331/q-2019-05-20-142},  
  title = {Subsystem symmetries, quantum cellular automata, and computational phases of quantum matter},
  author = {Stephen, David T. and Nautrup, Hendrik Poulsen and Bermejo-Vega, Juani and Eisert, Jens and Raussendorf, Robert},
  journal = {{Quantum}},
  issn = {2521-327X},
  publisher = {{Verein zur F{\"{o}}rderung des Open Access Publizierens in den Quantenwissenschaften}},
  volume = {3},
  pages = {142},
  month = may,
  year = {2019},
  eprint={1806.08780}
}

@article{nautrup2024measurement,
  title = {Measurement-based quantum computation from Clifford quantum cellular automata},
  author = {Poulsen Nautrup, Hendrik and Briegel, Hans J.},
  journal = {Phys. Rev. A},
  volume = {110},
  issue = {6},
  pages = {062617},
  numpages = {14},
  year = {2024},
  month = {Dec},
  publisher = {American Physical Society},
  doi = {10.1103/PhysRevA.110.062617},
  eprint={2312.13185}
}

@article{schlingemann2001quantum,
  title = {Quantum error-correcting codes associated with graphs},
  author = {Schlingemann, D. and Werner, R. F.},
  journal = {Phys. Rev. A},
  volume = {65},
  issue = {1},
  pages = {012308},
  numpages = {8},
  year = {2001},
  month = {Dec},
  publisher = {American Physical Society},
  doi = {10.1103/PhysRevA.65.012308},
  eprint={quant-ph/0012111}
}

@book{nielsen2010quantum,
  title={Quantum computation and quantum information},
  author={Nielsen, Michael A and Chuang, Isaac L},
  year={2010},
  isbn={978-1-107-00217-3},
  publisher={Cambridge university press},
  doi={10.1017/CBO9780511976667}
}

@article{gottesman1998heisenberg,
  title={The {H}eisenberg Representation of Quantum Computers},
  author={Gottesman, Daniel},
  journal={Proceedings of the XXII International Colloquium on Group Theoretical Methods in Physics (Group22)},
  editor={Corney, S. P. and Delbourgo, R. and Jarvis, P. D.},
  year={1999},
  publisher={International Press},
  address={Cambridge, MA},
  doi={10.48550/arXiv.quant-ph/9807006},
  eprint={quant-ph/9807006},
  archivePrefix={arXiv},
  primaryClass={quant-ph}
}

\appendix

\section{Rank-width of circle graphs}\label{app:rank_width_of_circle_graphs}

We note that it is possible to prove that there exist circle graphs with $n^2$ vertices and rank-width at least $cn$ (for a suitably chosen constant $c>0$) in a different, much less direct fashion using matroids.

First of all, it is known that the $n \times n$ grid has branch-width $n$; see~\cite{graphMinors10}. Moreover, for every graph with a cycle that is not a loop (such as the $n \times n$ grid), the branch-width of its cycle matroid is equal to its branch-width as a graph~\cite{branchwidthCycleMatroid}. Bouchet~\cite{connectivityIsotropic} proved a theorem about connectivity functions which implies that the branch-width of a binary matroid (such as the cycle matroid of a graph) is exactly one more than the rank-width of its fundamental graph. (For a more direct reference, see \cite[Proposition~3.1]{RWandVM}.) Finally, any fundamental graph of a planar graph (or equivalently, of its cycle matroid) is a circle graph by a theorem of de Fraysseix~\cite{Fraysseix1981}. Thus any fundamental graph of the $n \times n$ grid is a circle graph with $\mathcal{O}(n^2)$ vertices and rank-width at least $n-1$
(see e.g.~Figure~\ref{fig:planar_code_states_are_bipartite_circle_graphs} for the chord diagram associated to one of the fundamental graphs of the $4 \times 4$ grid).

Here, however, we opt to give a direct proof so that we do not need to rely on so many different results, and so that we can understand the comparability grid itself, which has been studied before in relation to rank-width~\cite{geelenGridTheoremVertexminors2023}.

For positive integers $m$ and $n$, the $m \times n$ \textit{comparability grid} is the simple graph with vertex set $\{1, 2, \ldots, m\} \times \{1, 2, \ldots, n\}$ where two distinct vertices $(i,j)$ and $(i',j')$ are joined by an edge if either both $i \leq i'$ and $j \leq j'$, or both $i \geq i'$ and $j \geq j'$.
(Note that simple graphs do not have loops, so for $(i,j)=(i',j')$ there is no edge.)

Thus a comparability grid is an example of a \textit{comparability graph}, i.e., a graph whose edges represent the comparable elements in some partially ordered set. Moreover, comparability grids are circle graphs; see \Cref{fig:compGrid} for the $3 \times 3$ comparability grid and one of its chord diagrams.

\begin{figure}
  \centering
  \includegraphics[width=0.8\columnwidth]{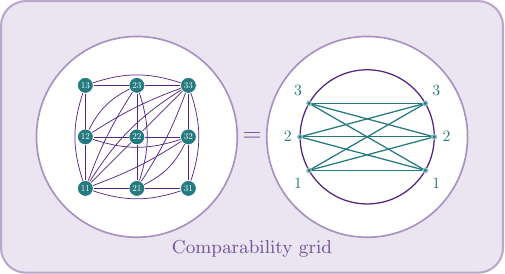}
  \caption{The $3 \times 3$ comparability grid (left) and a corresponding chord diagram (right).
  In the chord diagram, each comparability grid vertex $ij=(i, j)$ is represented by a chord between a point labeled $i$ in the left half and by a point labeled $j$ in the right half.
  Note that chords are considered to include their endpoints and thus chords with a shared endpoint intersect. }
  \label{fig:compGrid}
\end{figure}

To define rank-width, we fix some additional notation. Given a subset of vertices $X\subseteq V$ in a graph $G$, we write $r(X)$ for the rank over the binary field $\text{GF}(2)$ of the submatrix of the adjacency matrix with rows $X$ and columns~$V\setminus X$. We call $r(X)$ the \textit{cut-rank} of $X$. Note that $r(X) = r(V\setminus X)$ since the adjacency matrix is symmetric.

Next, a \textit{ternary tree} is a tree graph with only vertices of degree 1 or 3. The degree 1 vertices are called its \textit{leaves}. Every ternary tree $T$ with $n$ leaves gives rise to a \textit{rank-decomposition} of an $n$-vertex graph $G$ whose vertex set bijectively corresponds to the leaves of the ternary tree. Removing an edge $e$ from the tree $T$ cuts the tree into two subtrees and thus partitions both the leaf set of $T$ and thereby the vertex set of $G$ into two parts $(X_e, Y_e)$.

This partition $(X_e, Y_e)$ in turn assigns a \textit{width} to the edge $e$, which is the cut-rank $r(X_e)=r(Y_e)$. The \textit{width of the rank-decomposition} is then the maximum width of any of the edges of the tree. Finally, for any graph $G$, the \textit{rank-width} $rw(G)$ is the minimum width over all rank-decompositions of $G$, that is
\begin{equation*}\label{eq:rank-width}
    rw(G) = \min_{T} \max_{e \in E(T)} r(X_e).
\end{equation*}Thus a graph has small rank-width if, intuitively, it decomposes away along partitions of low cut-rank.

In order to prove that the comparability grid has large rank-width, we require the following lemma.

\begin{lemma}
\label{lem:smallerThanExpected}
Let $X$ contain at most half of the vertices of the $n \times n$ comparability grid and have cut-rank $r(X) \leq n/4$. Then $X$ contains strictly less than one-third of the vertices.
\end{lemma}
\begin{proof}
    We consider the vertices $(i,j)$ of the comparability grid to be arranged in a grid as indexed by $i$ and $j$, with $(1,1)$ at the bottom-left and $(n,n)$ at the top-right. We view $i$ as the $x$ coordinate and $j$ as the $y$ coordinate (see Figure~\ref{fig:compGrid}).
    Thus we refer to the \textit{rows} and \textit{columns}, where, for instance, the \textit{last} or \textit{$n$-th} column is $\{(n, j): j \in \{1,2,\ldots, n\}\}$.

    Choose a row or column with as many vertices in $X$ as possible. By symmetry between the rows and columns (this symmetry may be most clear by referring to either \Cref{fig:compGrid}), we may assume that it is a column, say column $i$. Let $k$ be the number of vertices in column $i$ which are in $X$. Going for a contradiction, suppose that $X$ contains at least a third of the vertices of the graph. So at least a third of column $i$ is in $X$, and thus $k \geq n/3$.

    \begin{figure}
  \centering
  \includegraphics[width=0.644\columnwidth]{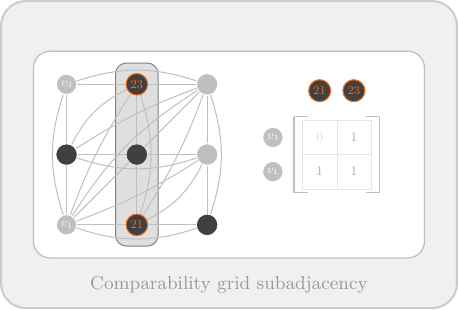}
  \caption{The comparability grid vertices in $X$ are shown in a \textcolor{black!75}{darker shade of gray} compared to \textcolor{gray!50}{the remaining vertices}.
  The square submatrix of the adjacency matrix whose rows are the left-holes $v_j$ and whose columns are the corresponding vertices $(i,j) \in X$ is shown on the right.
  }
  \label{fig:compGrid_holed}
\end{figure}

    We restrict ourselves to considering only the $k$ rows whose $i$-th vertex is in $X$. We say that such a row $j$ is \textit{left-holed} (respectively, \textit{right-holed}) if it contains a vertex $v_j$ to the left (respectively, right) of column $i$ which is not in $X$. We call this vertex $v_j$ a \textit{left-hole} (respectively, a \textit{right-hole}). We claim that at most $n/{4}$ rows are left-holed. To see this, consider the square submatrix of the adjacency matrix whose rows are the left-holes $v_j$ and whose columns are the corresponding vertices $(i,j) \in X$ (see \Cref{fig:compGrid_holed}).
    Since $v_j$ is adjacent to the $j$-th through $n$-th vertices in column $i$, this matrix is triangular with $1$s on the diagonal. Thus it has full rank. This proves the desired claim. A similar proof shows that at most $n/{4}$ rows are right-holed.

    Since $k\geq n/3>n/4$, not all of these $k$ rows are left-holed and not all are right-holed. So one of these $k$ rows has all vertices in column $i$ and further left in $X$, and one has all vertices in column $i$ and further right in $X$. Thus there exists a row with strictly more than half of its vertices in $X$. So in fact $k> n/2$. Thus there exists a row that is neither left-holed nor right-holed. So in fact $k=n$.

    Now, at least three-quarters of the vertices to the left of column $i$ are in $X$, and likewise for the right. This contradicts the fact that $X$ contains at most half the vertices.
\end{proof}

\begin{figure*}
  \centering
  \includegraphics[width=0.8\textwidth]{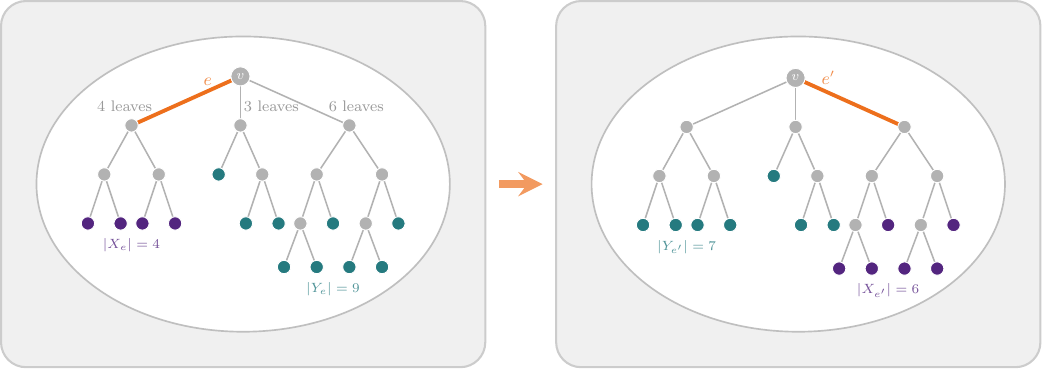}
  \caption{Finding a balanced partition.
  The degree-3 vertex $v$ adjacent to the edge $e$ on the $Y_e$ side has three incident branches, each with four, three, and six leaves, respectively.
  In the left box, the branch facing towards $X_e$ has strictly less than one third of the leaves ($\nicefrac{4}{13}<\nicefrac{1}{3}$).
  So at least one of the two branches facing towards $Y_e$ has strictly more than one third of the leaves; here, the one with six.
  Labeling the first edge of this branch as $e'$ gives a more balanced partition shown in the right box.
  }
  \label{fig:balanced_partition}
\end{figure*}

Now we can prove the main lemma.

\begin{lemma}
\label{lem:comparabilityGrid}
The rank-width of the $n \times n$ comparability grid is at least $n/{4}$.
\end{lemma}
\begin{proof}
    Going for a contradiction, suppose not. Let $T$ be a rank-decomposition of $G$ of width at most $n/4$. Let $e$ be an edge of $T$ so that $(X_e, Y_e)$ is as balanced as possible (i.e., so that $X_e$ and $Y_e$ are as close in size as possible).

    We may assume that $|X_e| \leq |Y_e|$ by simply switching the labels if  otherwise.

    Now, if $Y_e$ contained more than two-thirds of the vertices, then we could find an edge $e'$ corresponding to a more balanced partition $(X_{e'}, Y_{e'})$ by moving closer to $Y_e$. For an example of this, see \Cref{fig:balanced_partition}.

    Note that this is always possible because $T$ is a ternary tree. The degree-3 vertex $v$ adjacent to the edge $e$ on the $Y_e$ side has three incident branches partitioning the leaves of $T$. The branch facing towards $X_e$ has strictly less than one third of the leaves. So at least one of the two branches facing towards $Y_e$ has strictly more than one third of the leaves. Labeling the first edge of this branch as $e'$ gives a more balanced partition.
    So in fact $Y_e$ contains at most two thirds of the vertices and thus $X_e$ contains at least one third of the vertices (but also at most half since $Y_e$ is the larger side), i.e.~$n^2/3 \leq |X_e| \leq n^2/2$.

    But since $r(X_e)\leq n/4$, we can also apply \Cref{lem:smallerThanExpected} and conclude $|X_e| < n^2/3$, yielding a contradiction.
\end{proof}

Since comparability grids are a subset of circle graphs, there exist circle graphs with rank-width $\mathcal{O}(m)$ on $m^2$ vertices.

\begin{corollary}
\label{corollary:circle_graphs_rank-width}
Any upper bound on the rank-width of circle graphs must grow at least with $\mathcal{O}(\sqrt{n})$ in the number of vertices $n$.
\end{corollary}

We suspect that there actually exist circle graphs with rank-width at least $cn$ (for some constant $c>0$) on only $n$ vertices. The idea is to take chord diagrams which come from drawing the vertices of a $3$-regular expander on a circle and making the edges chords; this type of construction was discussed in the workshop \footnote{Potential constructions were discussed at the \href{https://bwag25.bici.events/home}{Bertinoro Workshop on Algorithms and Graphs} in 2025 following a question of Rutger Campbell.}. If we consider all graphs, not just circle graphs, then it is known that almost all graphs exhibit maximal rank-width, i.e. rank-width at least $\lceil n/3 \rceil - O(1)$~\cite{lee2012rank}.

It seems interesting to find more exact methods for lower-bounding the rank-width of a graph. Many other width parameters on graphs (such as treewidth~\cite{graphMinors10} and certain variants of tree-cut width~\cite{bramblesTreeCut}) have a ``dual'' object called \textit{brambles}. Brambles play a similar role to tangles (see~\cite{graphMinors10}), except they naturally have much more compact representations than tangles. It would be nice to find a natural notion of brambles for rank-width. Determining the exact rank-width of the $n \times n$ comparability grid could be a nice test case for the methods. We note that an abstract version of brambles has been considered in~\cite{bramblesSubmodular}, and it is known that polynomially-sized certificates of large rank-width do exist~\cite{certifyingLargeBW}.

\end{document}